\RequirePackage[british]{babel}
\documentclass[reqno,a4paper,12pt,final]{amsart}
\usepackage{fixltx2e}
\usepackage[utf8x]{inputenc}
\usepackage[T1]{fontenc}
\usepackage{tgpagella}
\usepackage[small]{eulervm}
\usepackage{amsmath,amssymb,amstext,amsthm,amscd,mathrsfs,eucal}
\usepackage[hmargin=2.7cm]{geometry}
\usepackage{graphicx,color}
\usepackage{array,colortbl,xtab}
\usepackage[all]{xy}
\usepackage[noadjust]{cite}
\usepackage{hyperref}
\hypersetup{%
  pdftitle   = {Half-BPS M2-brane orbifolds},
  pdfkeywords = {AdS/CFT, M2, M-theory, Killing spinors, 7-sphere quotients, orbifolds, Goursat},
  pdfauthor  = {Paul de Medeiros, José Figueroa-O'Farrill},
  pdfcreator = {\LaTeX\ with package \flqq hyperref\frqq}
}
\PrerenderUnicode{éÉ}
\newcommand{\one}{\boldsymbol{1}}
\newcommand{\half}{\tfrac12}

\newcommand{\Cl}{\mathrm{C}\ell}

\newcommand{\sA}{\mathbb{A}}
\newcommand{\sD}{\mathbb{D}}
\newcommand{\sE}{\mathbb{E}}

\newcommand{\fosp}{\mathfrak{osp}}

\newcommand{\PSO}{\mathrm{PSO}}
\newcommand{\SO}{\mathrm{SO}}
\newcommand{\Ort}{\mathrm{O}}
\newcommand{\Spin}{\mathrm{Spin}}
\newcommand{\Sp}{\mathrm{Sp}}
\newcommand{\SU}{\mathrm{SU}}

\newcommand{\RR}{\mathbb{R}}
\newcommand{\HH}{\mathbb{H}}

\newcommand{\CC}{\mathbb{C}}
\newcommand{\ZZ}{\mathbb{Z}}

\newcommand{\eN}{\mathscr{N}}

\newcommand{\bv}{\boldsymbol{v}}

\newcommand{\abar}{\overline{a}}
\newcommand{\bbar}{\overline{b}}

\newcommand{\Dbar}{\overline{D}}

\newcommand{\Fbar}{\overline{F}}
\newcommand{\Gbar}{\overline{\Gamma}}
\newcommand{\Abar}{\overline{A}}
\newcommand{\Bbar}{\overline{B}}
\newcommand{\alphabar}{\overline{\alpha}}
\newcommand{\betabar}{\overline{\beta}}
\newcommand{\lambdabar}{\overline{\lambda}}
\newcommand{\rhobar}{\overline{\rho}}
\newcommand{\that}{{\widehat\tau}}
\DeclareMathOperator{\AdS}{AdS}
\DeclareMathOperator{\Aut}{Aut}

\DeclareMathOperator{\Out}{Out}
\DeclareMathOperator{\Twist}{Twist}

\DeclareMathOperator{\Hom}{Hom}

\DeclareMathOperator{\id}{id}

\theoremstyle{plain}
\newtheorem{lemma}{Lemma}

\theoremstyle{definition}

\newtheorem*{problem}{Problem}
\newcommand{\MUNCH}[1]{\relax}

\allowdisplaybreaks
\begin{document}
\title[Half-BPS M2-brane orbifolds]{Half-BPS M2-brane orbifolds}
\author[de Medeiros]{Paul de Medeiros}
\author[Figueroa-O'Farrill]{José Figueroa-O'Farrill}
\address{School of Mathematics and Maxwell Institute for Mathematical Sciences, University of Edinburgh, Scotland, UK}
\email{\{P.deMedeiros,J.M.Figueroa\}@ed.ac.uk}
\date{\today}
\begin{abstract}
  Smooth Freund--Rubin backgrounds of eleven-dimensional supergravity of the form $\AdS_4 \times X^7$ and preserving at least half of the supersymmetry have been recently classified.  Requiring that amount of supersymmetry forces $X$ to be a spherical space form, whence isometric to the quotient of the round 7-sphere by a freely-acting finite subgroup of $\SO(8)$.  The classification is given in terms of ADE subgroups of the quaternions embedded in $\SO(8)$ as the graph of an automorphism.  In this paper we extend this classification by dropping the requirement that the background be smooth, so that $X$ is now allowed to be an orbifold of the round 7-sphere.  We find that if the background preserves more than half of the supersymmetry, then it is automatically smooth in accordance with the homogeneity conjecture, but that there are many half-BPS orbifolds, most of them new.  The classification is now given in terms of pairs of ADE subgroups of quaternions fibred over the same finite group.  We classify such subgroups and then describe the resulting orbifolds in terms of iterated quotients.  In most cases the resulting orbifold can be described as a sequence of cyclic quotients.
\end{abstract}
\maketitle
\tableofcontents
\listoftables

\section{Introduction}
\label{sec:introduction}

Recent advances in our detailed understanding of the $\AdS_4/\text{CFT}_3$ correspondence for (at least) half-BPS M2-brane configurations have been made possible by the explicit construction of superconformal field theories in three dimensions that are invariant under an orthosymplectic Lie superalgebra $\fosp( \eN |4)$ for $4 \leq \eN \leq 8$ \cite{BL2, BL3, MaldacenaBL, SchnablTachikawa, 3Lee, BHRSS, pre3Lee, GaiottoWitten, Imamura:2008dt}. Such theories preserve at least half the maximal amount of superconformal symmetry in three dimensions and the unitary theories have now been classified (see \cite{SCCS3Algs} and references therein for a comprehensive review). In several cases, the moduli spaces of gauge-inequivalent superconformal vacua for these theories have been analysed \cite{LambertTong, MaldacenaBL, ABJ, MasahitoBL, KlebanovBL, Terashima:2008ba, Imamura:2008nn, Imamura:2008ji, Imamura:2009ur, Imamura:2009ph} and found to contain a branch which, in the strong coupling limit, is identified with the expected dual geometry for the eight-dimensional space transverse to the worldvolume of a single M2-brane configuration preserving the same amount of supersymmetry. In each case, this geometry is found to describe an orbifold of the form $\RR^8 /\Gamma$, where $\Gamma$ is some finite subgroup of $SO(4) < SO(8)$ which commutes with the R-symmetry in the $\eN \geq 4$ conformal superalgebra. The orbifold $\RR^8 /\Gamma$ is precisely the cone over the quotient $S^7/\Gamma$ which appears in the dual Freund-Rubin solution $\AdS_4 \times S^7/\Gamma$ of eleven-dimensional supergravity in the near-horizon limit. Although the quotient $\RR^8 /\Gamma$ is necessarily singular (since the origin is always fixed by $\Gamma$), the quotient $S^7/\Gamma$ need not be. In fact, for all the known moduli spaces associated with M2-brane $\eN \geq 4$ superconformal field theories, the quotient $S^7/\Gamma$ is always smooth for $\eN >4$ whereas it is never smooth for $\eN =4$.

Arguably the most pressing open problem in this subject is to understand the precise nature of the dictionary between the superconformal field theories and their dual geometries.  As with natural languages, dictionaries contain two lists of words and a correspondence between them.  The two lists of words in the case of the $\AdS_4/\text{CFT}_3$ correspondence are the superconformal field theories on the one hand, and the geometries on the other.  As mentioned above the unitary superconformal Chern--Simons theories with $\eN\geq 4$ matter have been classified, and the purpose of this paper is to classify the possible dual geometries.

The mathematical problem we address in this paper is thus the classification of quotients $S^7/\Gamma$ which are spin and for which the vector space of real Killing spinors has dimension $\eN\geq 4$.  This boils down to the classification of certain subgroups of $\Spin(8)$ (up to conjugation).  In \cite{deMedeiros:2009pp} we discussed the case of smooth quotients and in this paper we would like to extend these results to include also singular quotients; that is, orbifolds.  We will see that those orbifolds for which $\eN>4$ are actually smooth, whence the list in \cite{deMedeiros:2009pp} is already complete.  The novelty in this paper is the determination of the $\eN=4$ orbifolds.

This paper is organised as follows.  In Section~\ref{sec:spherical-orbifolds} we discuss the geometry of spherical orbifolds in the context of supergravity backgrounds.   In Section~\ref{sec:spin-orbifolds} we review the fact that the orthonormal frame bundle of an orbifold is a smooth manifold, whence one can talk about spin structures just as for smooth manifolds, and we relate the Killing spinors on a spherical orbifold to the $\Gamma$-invariant parallel spinors on its cone orbifold.  In Section~\ref{sec:statement-problem} we rephrase the classification of $\eN\geq 4$ supersymmetric Freund--Rubin backgrounds of the form $\AdS_4 \times S^7/\Gamma$ as the classification of certain finite subgroups of $\Spin(8)$ up to conjugation.  For $\eN\geq 4$ supersymmetry, they are subgroups of $\Spin(4) \cong \Sp(1) \times \Sp(1)$.  The classification of finite subgroups of $\Spin(4)$ will take the first half of this paper.

The classification is an application of Goursat's theory of subgroups of a direct product of groups, which we review in Section \ref{sec:goursats-lemma} after introducing in Section~\ref{sec:finite-subgr-quat} the main ingredients of the classification: namely, the ADE subgroups of the quaternions.  In contemporary language, the main consequence of Goursat's theory is that finite subgroups of $\Sp(1) \times \Sp(1)$ are given as products of ADE subgroups $A$ and $B$ fibred over an abstract finite group $F$.  This means that $F$ is a common factor group of $A$ and $B$ by respective normal subgroups.  As explained in Section~\ref{sec:case-interest}, this requires classifying the normal subgroups of the ADE groups, their possible factor groups and their groups of outer automorphisms.  The determination of the normal subgroups and their corresponding factor groups is presented in Section~\ref{sec:quotients-ade}.  This is both well-known and well-hidden (or at least widely scattered) in the mathematical literature and we have had to recover these results ourselves.  We think it is probably useful to collect them here under one roof.

Section~\ref{sec:subgroups-spin4} puts everything together: we list the compatible pairs of ADE subgroups with isomorphic factor groups, work out the outer automorphisms, and assemble the fibred products.  The finite subgroups of $\Spin(4)$ are listed in three tables: Table~\ref{tab:products} contains the product groups (those fibred over the trivial group), Table~\ref{tab:smooth} contains the smooth quotients (those which are graphs of automorphisms), and Table~\ref{tab:remaining} lists the remaining groups.  The product groups require no effort and the smooth quotients were classified in \cite{deMedeiros:2009pp}, hence the main new result in this paper consists of Table~\ref{tab:remaining} and the subsequent analysis.  In Appendix~\ref{sec:struct-fibr-prod} we show that in almost all cases the isomorphism type of the subgroup only depends on the choice of ADE subgroups and their common factor group $F$ and \emph{not} on the outer automorphism of $F$ used to twist the product.  This result is perhaps not directly relevant to the subsequent analysis, but we found it useful to keep in mind, hence it is relegated to an appendix.  In Appendix~\ref{sec:finite-subgroups-so4} we recover the classification in \cite{MR1957212} of finite subgroups of $\SO(4)$ as an independent check on the results of the paper.

In Section~\ref{sec:expl-desc-orbif}, we make the classification more concrete by discussing the actual orbifolds.  We discuss how an orbifold by a group $\Gamma$ can be decomposed into orbifolds by smaller groups, starting from a normal subgroup of $\Gamma$ or, more generally, a subnormal series associated to $\Gamma$.  The theory behind this process is explained in Section~\ref{sec:orbi-iter-quots}.  We observe that with very few exceptions, orbifolds by finite subgroups of $\Spin(4)$ can be decomposed into a small number of iterated cyclic quotients.  The exceptions are the orbifolds which involve the binary icosahedral group, which we discuss briefly in Section~\ref{sec:orbif-involv-e8}.  The rest, which are the solvable groups, are treated in detail in Section~\ref{sec:solv-orbi-iter}.  We treat first the orbifolds by product groups, then the smooth quotients and finally the orbifolds by solvable groups in Table~\ref{tab:remaining}.  A few of the $\eN=4$ orbifolds have already appeared in the literature and in Section~\ref{sec:conclusion} we identify them and show where in our classification they occur.  The paper ends with a brief summary of the similar classification for the case of M5-branes.

\subsection*{How to use this paper}

The authors are the first to concede that the paper is somewhat technical, but we also believe that the results are potentially useful and hence we would like to offer the busy reader a brief user guide to the main results contained in the paper.

First a word about notation.  We distinguish between $\subset$, $<$ and $\lhd$.  Suppose $G$ is a group.  Then $S\subset G$ simply means that $S$ is a subset of $G$, whereas $S<G$ means that $S$ is a subgroup and $S\lhd G$ means it is a normal subgroup.  Normal subgroups play an important rôle in this paper, so let us say a few words about them.  Normal subgroups are the kernels of homomorphisms.  If $N\lhd G$ is a normal subgroup then $G/N$ becomes a group in such a way that the map $\pi: G \to G/N$ sending an element to its $N$-coset (left and right cosets agree for normal subgroups) is a group homomorphism with kernel $N$.  The relation between $G$, $N$ and $G/N$ can be succinctly summarised in terms of an exact sequence
\begin{equation}
  \begin{CD}
    1 @>>> N @>>> G @>\pi>> G/N @>>> 1~.
 \end{CD}
\end{equation}
It is important to stress that $G/N$ is an abstract group and not a subgroup of $G$.  If $G$ does have a subgroup $H$ to which $\pi$ restricts to give an isomorphism $\pi : H \stackrel{\cong}{\to} G/N$, then we say that the sequence splits, which implies that $G$ is the semidirect product $H \ltimes N$.  Of course, semidirect products include the direct products as a special case.  Exact sequences like the one above are also called group extensions and $G$ is said to be an extension of $G/N$ by $N$.  If the sequence splits, the extension is said to be trivial.

As explained in the body of the paper, $\eN\geq 4$ supersymmetry means that the relevant orbifolds are $S^7/\Gamma$, where $\Gamma$ is a subgroup of $\Sp(1) \times \Sp(1)$.  We work quaternionically, because we believe this makes the results much easier to describe, and the formulae much more natural.  This means that for us $\Sp(1)$ is indeed the group of unit quaternions and although it is isomorphic to $\SU(2)$ we eschew this isomorphism; although for the sake of comparison with the existing literature we provide a brief dictionary in Section~\ref{sec:conclusion}.  The action of $\Sp(1) \times \Sp(1)$ on $S^7$ is defined by letting $S^7$ be the unit sphere in $\HH^2$ and having $\Sp(1) \times \Sp(1)$ act on $\HH^2 = \HH \oplus \HH$ by left multiplication: namely, $(u,v) \in \Sp(1)\times \Sp(1)$ sends $(x,y)\in \HH^2$ to $(ux,vy)$.  This defines an embedding of $\Sp(1) \times \Sp(1)$ into $\SO(4) \times \SO(4)$ and hence into $\SO(8)$.

Two Freund--Rubin backgrounds $\AdS_4 \times S^7/\Gamma$ and $\AdS_4 \times S^7/\Gamma'$ are equivalent if and only if $\Gamma$ and $\Gamma'$ are conjugate in $\Spin(8)$.  This freedom allows us to make a number of choices along the way.  The upshot of the analysis is that each subgroup $\Gamma$ is equivalent to one which is uniquely characterised by a 4-tuple $(A,B,F,\tau)$, which we now explain.

First of all, $A,B$ are ADE subgroups of $\Sp(1)$.  These are very concrete finite subgroups of $\Sp(1)$ described in Table \ref{tab:ADE} along with their explicit quaternion generators.  We use the Dynkin labels $\sA_{n-1}$, $\sD_{n+2}$, $\sE_{6,7,8}$ (for $n\geq 2$) to refer to those precise subgroups and \emph{not} to their isomorphism classes: $\ZZ_n$, $2D_{2n}$, $2T$, $2O$ and $2I$, respectively.

The group $F$ is an \emph{abstract} finite group (i.e., not a subgroup of anything), admitting surjections from both $A$ and $B$.  Since the kernel of the surjections $A\to F$ and $B\to F$ are normal subgroups of $A$ and $B$, respectively, it is enough to classify the possible normal subgroups of the ADE subgroups of $\Sp(1)$.  They are summarised, along with the corresponding $F$s, in Table~\ref{tab:normal}.

Finally, $\tau \in \Out(F)$ is a representative of a certain equivalence class of outer automorphisms of $F$.  This requires determining the group $\Aut(F)$ of automorphisms of $F$, the group of inner automorphisms and hence the group $\Out(F)$ of outer automorphisms.  The results are summarised in Table~\ref{tab:outer}.  In some cases, different outer automorphisms give rise to equivalent orbifolds.  The equivalence relation agrees with that defined by the action on $\Out(F)$ by a certain group.  The set of orbits we call the set of possible twists and denote it $\Twist(F)$; although it does not just depend on $F$ but also on how $F$ is obtained as a factor of $A$ and $B$.

Out of the data $(A,B,F,\tau)$ one defines a subgroup $A\times_{(F,\tau)} B$ of $A\times B$ called a (twisted) fibred product.  There are many such subgroups and are given in Tables~\ref{tab:products}, \ref{tab:smooth} and \ref{tab:remaining}, along with their orders.  It turns out that in most cases the isomorphism type of $A\times_{(F,\tau)} B$ is independent of $\tau$, a fact we prove in Appendix~\ref{sec:struct-fibr-prod}.  We do not work out the isomorphism type in all cases, since this takes some effort and it is not clear how useful this information actually is.  The three tables just mentioned contain the list of finite subgroups of $\Sp(1) \times \Sp(1)$, up to the action of the automorphism group $\Aut(\Sp(1) \times \Sp(1))$, which is an extension of the adjoint group $\SO(3) \times \SO(3)$ of inner automorphisms by the $\ZZ_2$-group generated by interchanging the two $\Sp(1)$ factors.

Although the meaning of $A\times_{(F,\tau)} B$ is explained formally in Section~\ref{sec:gours-theory-subgr}, perhaps it is convenient to describe it here more informally.  This is how, in practice, we work with such fibred products.  The fact that $A$ and $B$ have $F$ as a common factor group, means that there exist normal subgroups $A_0 \lhd A$ and $B_0 \lhd B$ such that there is a group isomorphism $A/A_0 \cong B/B_0$.  Once such an isomorphism has been chosen, the fibred product $A\times_{(F,\tau)} B$ is given by the preimage of the graph of $\tau$ in $F\times F$ under the map $\pi$ in the exact sequence
\begin{equation}
  \begin{CD}
    1 @>>> A_0 \times B_0 @>>> A \times B @>\pi>> F \times F @>>> 1.
  \end{CD}
\end{equation}
It consists of pairs $(a,b) \in A \times B$ such that $bB_0 = \tau(aA_0)$, where we write $aA_0$ and $bB_0$ for the $A_0$-coset of $a$ and the $B_0$-coset of $b$, respectively.  This description of $A\times_{(F,\tau)} B$ allows us to deconstruct the corresponding orbifold as a sequence of quotients: first we quotient by $A_0 \times B_0$ and the we quotient the resulting orbifold by $F$.  The quotients by $A_0 \times B_0$ and $F$ can themselves be decomposed into simpler quotients, et cetera.  In the end, all but a handful of orbifolds can be described as a small number of iterated cyclic quotients.  The results can be read off of several tables: Table~\ref{tab:productsiterated} for the product groups, Table \ref{tab:smoothiterated} for the smooth ``orbifolds'' and Table~\ref{tab:remainingiterated} for the remaining groups.  The notation in those tables is explained in the bulk of the paper, but we summarise it here for the impatient:  $\omega_n = e^{2\pi i/n}$ (but we use $\xi = \omega_8$) and $\zeta = e^{i\pi/4}e^{j\pi/4}$.  We give the generators of the cyclic groups we quotient by \emph{not} as abstract generators, but as explicit elements of $\Sp(1) \times \Sp(1)$.

For example, let us consider the orbifold $S^7/\Gamma$, where $\Gamma = \sE_6 \times_T \sE_6$ has order 48.  We can perform this quotient in steps.  We first consider the quotient $X$ of $S^7$ by the $\ZZ_2 \times \ZZ_2$ subgroup of $\Sp(1)\times\Sp(1)$ generated by $(-1,1)$ and $(1,-1)$.  The element $(i,i)$ acts on $X$ with order $2$, since $(i,i)^2 = (-1,-1)$ and this is an element of the $\ZZ_2 \times \ZZ_2$ group we  quotiented $S^7$ by in order to obtain $X$.  Let $Y$ be the quotient of $X$ by $(i,i)$.  On $Y$ we have an action of $(j,j)$, which again has order $2$.  Let the corresponding quotient be $Z$. Finally, on $Z$ we have the action of $(\zeta,\zeta)$ which has order $3$, since $(\zeta,\zeta)^3 = (-1,-1)$.  That $\ZZ_3$-quotient of $Z$ is $S^7/\Gamma$.  What this cyclic decomposition of the group $\Gamma$ is doing is basically allowing us to write the elements of $\Gamma$ uniquely as a word in some cyclic generators \emph{with a chosen order} of those generators, and this suggests performing the quotient by each such generator in turn.

\section{Spherical orbifolds}
\label{sec:spherical-orbifolds}

Let $S^7$ denote the unit sphere in $\RR^8$.  The Lie group of orientation-preserving isometries of $S^7$ is the special orthogonal group $\SO(8)$.  Let $\Gamma < \SO(8)$ be a finite group.  The quotient $S^7/\Gamma$, obtained by identifying points on the sphere which are on the same orbit of $\Gamma$, will be smooth if and only if $\Gamma$ acts freely; that is, with trivial stabilisers.  Recall that the stabilizer $\Gamma_x$ of a point $x \in S^7$ is the subgroup of $\Gamma$ consisting of all the elements of $\Gamma$ which fix $x$.  The action of $\Gamma$ is free if $\Gamma_x = \{1\}$ for all $x \in S^7$.  The determination of the smooth quotients of $S^7$ (or, more generally, any round sphere) is the so-called spherical space-form problem, which has a long history culminating in Wolf's classification \cite{Wolf}.  In this paper we shall be concerned with quotients which are not necessarily smooth; equivalently, with finite subgroups of $\SO(8)$ which do not act freely.  We shall call them orbifolds; although in the mathematical literature they are more precisely known as \emph{global} orbifolds.

\subsection{Spin orbifolds}
\label{sec:spin-orbifolds}

An often under-appreciated fact is that the bundle of oriented orthonormal frames of an orbifold is actually a smooth manifold.  For the case of $S^7/\Gamma$ (or indeed any orbifold of a round sphere) we can see this explicitly.  First of all, we observe that the bundle of oriented orthonormal frames $P_{\SO}(S^7)$ of $S^7$ is diffeomorphic to the Lie group $\SO(8)$.  To see this let $x \in S^7$: it is a unit-norm vector in $\RR^8$.  The tangent space to $S^7$ at $x$ is the hyperplane in $\RR^8$ perpendicular to $x$.  An orthonormal frame for $T_xS^7$ is a set of seven unit-norm vectors in $\RR^8$ which are perpendicular to $x$ and to each other.  Together with $x$ they make up an orthonormal frame for $\RR^8$, whence the $8 \times 8$ matrix with these vectors as columns and $x$ in the first column, say, is orthogonal.  An orientation on $S^7$ is a choice of sign for the determinant of this orthogonal matrix; equivalently, it is a choice of connected component in the orthogonal group $\Ort(8)$.  Either choice results in a manifold diffeomorphic to $\SO(8)$.  Let us, for definiteness, choose the orientation corresponding to unit determinant, so that $P_{\SO}(S^7) \cong \SO(8)$.  The projection $\SO(8) \to S^7$ is the map which selects the first column of the matrix.

A finite subgroup $\Gamma<\SO(8)$ acts linearly on $\RR^8$ and the action on $S^7$ is the restriction of this action.  Its lift to the tangent bundle is also a linear action.  In other words, under the identification $P_{\SO}(S^7) = \SO(8)$, the action of $\gamma$ on $P_{\SO}(S^7)$ is just left matrix multiplication in $\SO(8)$, which is a free action.  Furthermore this action is compatible with the bundle projection $P_{\SO}(S^7) \to S^7$.  The space of right $\Gamma$-cosets in $\SO(8)$, which is a smooth manifold, can be identified with the bundle of orthonormal frames $P_{\SO}(S^7/\Gamma)$ of the orbifold.  Notice that it is \emph{not} a principal $\SO(7)$ bundle: if it were, the right action of $\SO(7)$ would be free and the quotient a manifold instead of an orbifold.

The sphere $S^7$ has a unique spin structure.  Indeed, the total space of the spin bundle is diffeomorphic to the Lie group $\Spin(8)$ and the covering homomorphism $\Spin(8) \to \SO(8)$ is the bundle morphism from the spin bundle $P_{\Spin}(S^7)$ to $P_{\SO}(S^7)$.  Now the orbifold $S^7/\Gamma$ admits a spin structure if and only if the action of $\Gamma$ on $P_{\SO}(S^7)$ lifts to $P_{\Spin}(S^7)$ in a way that is compatible with the bundle map $P_{\Spin}(S^7) \to P_{\SO}(S^7)$.  This is equivalent to the existence of a lift of $\Gamma$ to an \emph{isomorphic} subgroup (also denoted $\Gamma$) of $\Spin(8)$ (as opposed to a double cover), which acts by left multiplication on $\Spin(8)$.  This action is again free and we define the spin bundle of the orbifold $S^7/\Gamma$ to be the smooth manifold of right $\Gamma$-cosets in $\Spin(8)$.  If we think of $\Spin(8)$ as sitting inside the Clifford algebra $\Cl(8)$, then the condition that $\Gamma$ map isomorphically onto its image in $\SO(8)$ is that $-\one \not\in \Gamma$.  There are topological obstructions to the existence of such lifts and even when the obstruction is overcome there may be more than one lift, each one corresponding to a different spin structure on the orbifold.  These are classified, as in the case of manifolds, by $\Hom(\Gamma,\ZZ_2)$, which boils down to introducing signs for each generator of $\Gamma$ in a way consistent with the relations.

Bär's cone construction \cite{Baer} relates Killing spinors on $S^7$ to parallel spinors on $\RR^8$.  More precisely the vector space of Killing spinors (with Killing constant $\half$, say) on $S^7$ is isomorphic to the space of parallel positive-chirality spinors on $\RR^8$.  Furthermore this correspondence is equivariant with respect to the action of $\Spin(8)$, as shown more generally in \cite{JMFKilling} in a similar context to that of the present paper. Relative to flat coordinates for $\RR^8$, a parallel spinor is a constant spinor, and hence the space of Killing spinors is isomorphic, as a representation of $\Spin(8)$, to $\Delta_+$, the positive-chirality spinor irreducible representation of $\Spin(8)$.  The subgroup $\Gamma < \Spin(8)$ acts naturally on this representation and the Killing spinors on $S^7/\Gamma$ can be identified with the $\Gamma$-invariant spinors in $\Delta_+$.  They form a vector space $\Delta_+^\Gamma$ of dimension $0\leq \eN \leq 8$.

\subsection{Statement of the problem}
\label{sec:statement-problem}

We are interested in classifying orbifolds $S^7/\Gamma$, up to isometry, admitting real Killing spinors (and hence a spin structure) and such that the (real) dimension $\eN$ of the vector space of real Killing spinors be $\geq 4$.  This is equivalent to classifying subgroups $-\one \not\in \Gamma < \Spin(8)$, up to conjugation, such that $\dim\Delta_+^\Gamma \geq 4$.  We shall call such subgroups $\Gamma$ \emph{admissible} in this paper.  It is important to stress that it is conjugation in $\Spin(8)$ that we have to consider as our basic equivalence relation and not the weaker conjugation of $\Gamma$ in $\SO(8)$.  Certainly quotients by conjugate subgroups of $\SO(8)$ are isometric, but this is not enough to guarantee that their spin bundles are also isomorphic, and hence that the dimension of the space of real Killing spinors be equal.  This is consistent with the fact that supergravity backgrounds are not merely orbifolds \emph{admitting} a spin structure, but orbifolds with a choice of spin structure on them.  It is easy to see that if $\Gamma$ and $\Gamma'= g \Gamma g^{-1}$, for some $g \in \Spin(8)$, are admissible subgroups, then $S^7/\Gamma$ and $S^7/\Gamma'$ are isomorphic as spin orbifolds and hence the corresponding Freund--Rubin backgrounds are equivalent.

Let $\Gamma < \Spin(8)$ be an admissible subgroup.  Let $\delta_+: \Spin(8) \to \SO(\Delta_+)$ denote the chiral spinor representation of $\Spin(8)$.  Then since $\dim \Delta_+^\Gamma =\eN$, the image of $\Gamma$ under $\delta_+$ is contained in an $\SO(8-\eN)$ subgroup of $\SO(\Delta_+)$ corresponding to the perpendicular complement $(\Delta_+^\Gamma)^\perp$ of the subspace of $\Gamma$-invariant spinors.  This means that $\Gamma$ is contained in a $\Spin(8-\eN)$ subgroup of $\Spin(8)$.  Now $\Spin(8)$ acts transitively on the grassmannian of $\eN$-planes in $\Delta_+$, so we can use the freedom to conjugate $\Gamma$ in $\Spin(8)$ to ensure that $\Gamma$ belongs to a particular $\Spin(8-\eN)$ subgroup.  (This does not mean that all $\Spin(8-\eN)$ subgroups of $\Spin(8)$ are conjugate, just those which leave pointwise invariant an $\eN$-dimensional subspace of $\Delta_+$.)

Let us consider, for definiteness, the case $\eN=4$, since $\eN>4$ will just be a specialisation of this case.  We will fix some of the freedom to conjugate in $\Spin(8)$ by the requirement that $\Gamma$ belongs to a particular $\Spin(4)$ subgroup of $\Spin(8)$.  We choose this subgroup to be such that its action on $S^7$ is as follows.  First notice that $\Spin(4) \cong \Sp(1) \times \Sp(1)$, with $\Sp(1)$ the Lie group of unit-norm quaternions.  This group acts naturally on $\HH^2 \cong \HH \oplus \HH$ via left quaternion multiplication:
\begin{equation}
  (u_1,u_2) \cdot (x_1,x_2) = (u_1 x_1, u_2 x_2)~,
\end{equation}
where $u_i\in\Sp(1)$ and $x_i \in \HH$, in such a way that it preserves the unit sphere $S^7 \in \HH^2$.  One can check that this $\Spin(4)$ leaves invariant pointwise a four-plane in $\Delta_+$.  Once we conjugate $\Gamma$ to lie inside this $\Spin(4)$ subgroup, we still have the freedom to conjugate by the normaliser in $\Spin(8)$ of this $\Spin(4)$ subgroup.  Of course, the normaliser contains the $\Spin(4)$ subgroup itself, but it also contains other elements of $\Spin(8)$ inducing non-inner automorphisms of the $\Spin(4)$ subgroup.  As shown, e.g., in \cite[§3]{FSS3}, there is a unique nontrivial outer automorphism of $\Spin(4) \cong \Sp(1)\times \Sp(1)$, and it is represented by the automorphism which interchanges the two $\Sp(1)$ subgroups: $(u_1,u_2) \mapsto (u_2, u_1)$.  Indeed, since the Lie algebras of the $\Sp(1)$ subgroups correspond to the self-dual and anti-self-dual 2-forms in $\Lambda^2 (\Delta_+^\Gamma)^\perp$, all we need to do to interchange them is to change the orientation of $(\Delta_+^\Gamma)^\perp$.  This can be done by a rotation of $\pi$ degrees in a 2-plane spanned by a spinor in $(\Delta_+^\Gamma)^\perp$ and a spinor in $\Delta_+^\Gamma$ and such a rotation is induced by conjugating by an element in $\Spin(\Delta_+)$, which becomes an element in $\Spin(8)$ after a triality transformation.

In summary, we want to classify finite subgroups $\Gamma$ of $\Sp(1) \times \Sp(1)$ up to conjugation in $\Sp(1) \times \Sp(1)$ and the outer automorphism which swaps the two $\Sp(1)$ subgroups; in other words, we want to solve the following

\begin{problem}
  Classify the finite subgroups of $\Sp(1) \times \Sp(1)$ up to automorphisms of $\Sp(1) \times \Sp(1)$.
\end{problem}

This will result from an application of Goursat's theory of subgroups of a direct product of groups, but before doing so, we review the finite subgroups of $\Sp(1)$, which will be the main ingredients in terms of which our results will be phrased.

\section{Finite subgroups of the quaternions}
\label{sec:finite-subgr-quat}

In this section we review very briefly, mostly to settle the notation and to make this paper reasonably self-contained, the classification of finite subgroups of the quaternions in terms of simply-laced (extended) Dynkin diagrams.

Let $\HH$ denote the skew field of quaternions.  Multiplication in $\HH$ being associative, we can talk about multiplicative subgroups.  Let $G < \HH$ be a finite such subgroup.  Since $G$ is a finite group, every $x \in G$ obeys $x^n =1$ for some finite $n$ and since $\HH$ is a normed algebra, it follows that $|x|^n = |x^n| = 1$, whence $|x|=1$.  In other words, $G$ is actually a subgroup of $\Sp(1)$, the Lie group of unit-norm quaternions.  The adjoint representation of $\Sp(1)$ defines a covering homomorphism $\Sp(1) \to \SO(3)$, once we choose an orthonormal basis for the Lie algebra of $\Sp(1)$, which is the three-dimensional real vector space of imaginary quaternions.  Therefore every finite subgroup $G < \Sp(1)$ projects to a finite subgroup $\overline{G} < \SO(3)$.

The classification of finite subgroups of $\SO(3)$ is classical and a nice treatment can be found in Elmer Rees's \emph{Notes on geometry} \cite{MR681482}.  A result due to Euler states that any nontrivial rotation in $\RR^3$ is the rotation about some axis.  That axis intersects the unit sphere in $\RR^3$ at two points: the \emph{poles} of the rotation.  Let $P$ denote the set of all the poles of the non-identity elements in $\overline{G}$.  The finite group $\overline{G}$ acts on the finite set $P$ and a careful application of the orbit-stabilizer theorem shows that it does so with either two or three orbits.  The case with two orbits corresponds to a cyclic group, where the two orbits correspond to the two poles common to all the rotations.  The case with three orbits corresponds either to a dihedral group or to the group of symmetries of a Platonic solid: namely, the tetrahedral, octahedral or icosahedral groups.

One then determines the possible lifts of these finite rotation subgroups to $\Sp(1)$, thus determining the finite subgroups of quaternions.  They turn out to be classified by the ADE Dynkin diagrams and tabulated in Table~\ref{tab:ADE}.  They are given by cyclic groups, binary dihedral groups and binary polyhedral groups.  In the rows labelled $\sA_{n-1}$ and $\sD_{n+2}$ we take $n\geq 2$.  In the last row, $\phi = \half(1+\sqrt{5})$ is the Golden Ratio.  It should be stressed that the labels $\sA$,$\sD$ and $\sE$ refer to the explicit subgroups in this table, with the generators shown and not just to the isomorphism class of such subgroups.

\begin{table}[h!]
  \caption{Finite subgroups of $\Sp(1)$}
  \centering
  \begin{tabular}{>{$}l<{$}|>{$}c<{$}|>{$}c<{$}|>{$}l<{$}}
    \multicolumn{1}{c|}{Dynkin} & \text{Name} & \text{Order} & \text{Quaternion generators}\\\hline
    \sA_{n-1} & \ZZ_n & n & e^{\frac{i2\pi}{n}}\\
    \sD_{n+2} & 2D_{2n} & 4n & j, e^{\frac{i\pi}{n}}\\
    \sE_6 & 2T & 24 & \frac{(1+i)(1+j)}{2},\frac{(1+j)(1+i)}{2}\\
    \sE_7 & 2O & 48 & \frac{(1+i)(1+j)}{2},\frac{1+i}{\sqrt{2}}\\
    \sE_8 & 2I & 120 & \frac{(1+i)(1+j)}{2},\frac{\phi + \phi^{-1}i + j}{2}
  \end{tabular}
  \label{tab:ADE}
\end{table}

The Dynkin label is due to the McKay correspondence \cite{MR604577} and goes as follows.  We associate a graph to a finite subgroup $G < \Sp(1)$ as follows.  Since $G$ is finite, there are a finite number of irreducible representations: $V_0,V_1,\dots,V_r$.  In addition there is a ``fundamental'' representation $R$ (not necessarily irreducible) obtained by restricting the two-dimensional complex representation of $\Sp(1)$ given by the isomorphism $\Sp(1) \cong \SU(2)$.  The vertices of the graph are labelled by the irreducible representations of $G$ and we draw an edge between vertices $V_i$ and $V_j$ if and only if $V_j$ appears in the decomposition into irreducibles of the tensor product representation $V_i \otimes R$.  It turns out that this incidence relation is symmetric, since $V_i$ appears in $V_j \otimes R$ if and only if $V_j$ appears in $V_i \otimes R$.  The resulting graph is therefore not directed.  It can be shown to be an extended Dynkin diagram in the ADE series.  Deleting the vertex corresponding to the trivial representation, as well as any edge incident on it, we obtain the corresponding Dynkin diagram.

\section{Goursat's Lemma}
\label{sec:goursats-lemma}

The determination of the finite subgroups of $\Sp(1) \times \Sp(1)$ follows from work of Goursat \cite{Goursat}. In fact, Goursat's original work was motivated by the determination of finite subgroups of $\SO(4)$ and he seems to have solved this problem (modulo a few missing cases) by noticing that $\SO(4)$ is covered by a direct product group, namely $\Sp(1) \times \Sp(1)$, and determining the subgroups of such a direct product.  This ought to make it possible to read the results from Goursat's paper, but this is a serious undertaking comparable in magnitude to (but significantly less fun than) recovering his results independently departing from his basic idea: the (algebraic) Goursat Lemma, which we discuss presently.  In fact, there is a further complication in that, according to \cite{MR1957212}, Goursat's classification is incomplete, as is Du Val's \cite{MR0169108}.  This is remedied in \cite{MR1957212}.  In Appendix~\ref{sec:finite-subgroups-so4} we recover the classification of finite subgroups of $\SO(4)$ (up to the action of the automorphisms of $\SO(4)$).  Since our result agrees with the classification in \cite{MR1957212}, this provides an independent check of both results.

\subsection{Goursat's theory of subgroups}
\label{sec:gours-theory-subgr}

The question is to determine (at least in principle) the subgroups of a direct product of groups.   The short answer is that subgroups are in bijective correspondence with graphs of isomorphisms between factor groups of the groups in question.  Let us elaborate.

Let $A,B$ be groups and let $C < A \times B$ be a subgroup.  We have canonical group homomorphisms $\lambda: C \to A$ and $\rho: C \to B$ obtained by restricting the cartesian projections $\pi_1: A \times B \to A$ and $\pi_2: A \times B \to B$ to $C$.  In fact, it may be convenient to think of $C$ not as a subgroup of $A \times B$ but as an injective homomorphism $\iota : C \to A \times B$ and the maps $\lambda,\rho$ as its compositions with $\pi_1,\pi_2$, respectively.  In summary, we have the following commutative diagram of groups:
\begin{equation}
  \label{eq:diagram}
  \xymatrix{
    & C \ar@/_/[ddl]_\lambda \ar@/^/[ddr]^\rho \ar[d]^\iota & \\
    & A \times B \ar[dl]^{\pi_1} \ar[dr]_{\pi_2} & \\
    A & & B
  }
\end{equation}
Clearly if $C$ is a subgroup of $A \times B$, it's actually a subgroup of $\lambda(C) \times \rho(C)$, whence without any loss of generality we can and will assume that $\lambda: C \to A$ and $\rho: C \to B$ are surjective.  The kernels of $\lambda$ and $\rho$ define normal subgroups of $C$ and since the image of a normal subgroup under a surjective homomorphism is again normal, we get normal subgroups $A_0 := \lambda(\ker\rho)$ of $A$ and $B_0 := \rho(\ker\lambda)$ of $B$.  In terms of elements,
\begin{equation}
  \label{eq:normal}
  A_0 = \left\{a \in A \middle | (a,1) \in C \right\} \qquad\text{and}\qquad   B_0 = \left\{b \in B \middle | (1,b) \in C \right\}~,
\end{equation}
where $1$ denotes the identity element --- after all we are eventually interested in subgroups of the quaternions, where the identity  element is indeed the number $1$ --- and where we have identified $C$ with its image under $\iota$ in $A \times B$.

Goursat's Lemma states that the factor groups $A/A_0$ and $B/B_0$ are isomorphic, the isomorphism being given essentially by $C$.  Let us define a map $\overline\varphi: A \to B/B_0$ using the elements of $C$: if $(a,b) \in C$, we define $\overline\varphi(a) = b B_0$.  This map is well-defined, because if $(a,b_1)$ also belongs to $C$, then so do $(a,b_1)^{-1} = (a^{-1},b_1^{-1})$ and the product $(a,b_1)^{-1}(a,b) = (1, b_1^{-1} b)$, whence $b_1^{-1} b \in B_0$ and hence $b_1 B_0 = b_1 b_1^{-1} b B_0 = bB_0$.  The map $\overline\varphi$ is surjective, since by definition for every $b \in B$, there is some $a \in A$ so that $(a,b) \in C$.  The kernel of $\overline\varphi$ consists of those $a \in A$ with $(a,b_0) \in C$ for some $b_0 \in B_0$.  This means that $(1,b_0) \in C$, and hence $(a,b_0)(1,b_0)^{-1} = (a,1) \in C$, whence $a \in A_0$.  Conversely if $a \in A_0$, then $(a,1) \in C$ and hence $a \in \ker\overline\varphi$.  Therefore $\ker \overline\varphi = A_0$ and hence $\overline\varphi$ induces an isomorphism $\varphi: A/A_0 \xrightarrow{\cong} B/B_0$, sending $aA_0 \mapsto bB_0$ for $(a,b) \in C$.

Let $F$ be an abstract group isomorphic to both $A/A_0$ and $B/B_0$ and let us choose, once and for all, isomorphisms $A/A_0 \xrightarrow{\cong} F$ and $B/B_0 \xrightarrow{\cong} F$.  Then the isomorphism $\varphi: A/A_0 \to B/B_0$ defined by $C$ induces an automorphism of $F$.

In summary, $C$ determines (and is determined by) the following data:
\begin{enumerate}
\item groups $A$ and $B$,
\item normal subgroups $A_0 \lhd A$ and $B_0 \lhd B$ with $A/A_0 \cong B/B_0 \cong F$, and
\item an automorphism of $F$.
\end{enumerate}

To reconstruct $C$ from the above data, consider the following exact sequence
\begin{equation}
  \begin{CD}
    1 @>>> A_0 \times B_0 @>>> A \times B @>>> F \times F @>>> 1,
  \end{CD}
\end{equation}
obtained canonically from the data above, including the isomorphisms $A/A_0 \cong F$ and $B/B_0 \cong F$.  Now consider the subgroup $F_\tau$ of $F \times F$ given by the graph of an automorphism $\tau: F \to F$.  Explicitly, $F_\tau$ is the subgroup of $F \times F$ consisting of elements $\left\{(x,\tau(x)) \middle | x \in F\right\}$ and isomorphic to $F$.  (The isomorphism is projecting onto the first factor, for example.) The subgroup $C$ of $A\times B$ is then the subgroup which maps to $F_\tau$ under the map $A\times B \to F \times F$.  Since $A_0 \times B_0$ is the kernel of this map, we see that $A_0 \times B_0$ is contained in $C$ as a normal subgroup.   In other words, $C$ is an extension of $F_\tau$ by $A_0 \times B_0$:
\begin{equation}
  \begin{CD}
    1 @>>> A_0 \times B_0 @>>> C @>>> F_\tau @>>> 1.
  \end{CD}
\end{equation}

Equivalently, $C$ can be interpreted as a fibred product, also known as a categorical pullback.  To see this let us fix isomorphisms $i_A: A/A_0 \to F$ and $i_B: B/B_0 \to F$ once and for all.  Let $\tau \in \Aut(F)$ be an automorphism.  Let $\beta: B \to F$ denote the composition
\begin{equation}
  \label{eq:beta}
  \xymatrix{ B \ar[r] & B/B_0 \ar[r]^{i_B} & F }
\end{equation}
and $\alpha : A \to F$ denote the composition
\begin{equation}
  \label{eq:alpha}
  \xymatrix{ A \ar[r] & A/A_0 \ar[r]^{i_A} & F \ar[r]^\tau & F~,}
\end{equation}
where the maps $A \to A/A_0$ and $B \to B/B_0$ are the canonical ones.  It is clear that both $\alpha,\beta$ are surjective.  Then the subgroup $C < A \times B$ is the pullback
\begin{equation}
  \label{eq:pullback}
  \xymatrix{
    C \ar[r]^\rho \ar[d]_\lambda & B \ar[d]^\beta \\
    A \ar[r]^\alpha & F
    }
\end{equation}
or in terms of elements
\begin{equation}
  \label{eq:elements}
  C = \left\{(a,b) \in A \times B \middle | \alpha(a) = \beta(b)\right\}~,
\end{equation}
which one recognises as the fibred product $A \times_F B$.  It must be stressed that although the notation may not reflect it, the data defining the fibred product $A\times_F B$ are not just the groups $A,B,F$, but indeed the surjections $\alpha: A\to F$ and $\beta: B \to F$.  In particular, $\alpha$ incorporates the automorphism $\tau$ of $F$.  We therefore reluctantly introduce the notation $A \times_{(F,\tau)} B$ to reflect this.  We will say that the fibred product in this case is \emph{twisted} by $\tau$.  In Appendix~\ref{sec:struct-fibr-prod} we will prove that for almost all of the fibred products in this paper, twisting results in abstractly isomorphic groups.

\subsection{The case of interest}
\label{sec:case-interest}

Let us now consider the case at hand: classifying finite subgroups $\Gamma < \Sp(1) \times \Sp(1)$ up to automorphisms of $\Sp(1)\times\Sp(1)$.  Such subgroups will be determined by a pair of finite subgroups $A < \Sp(1)$ and $B < \Sp(1)$ having isomorphic factor groups \emph{and} the explicit isomorphism between the two factor groups.  Equivalently, $\Gamma$ is determined by the following data:
\begin{enumerate}
\item finite subgroups $A < \Sp(1)$ and $B < \Sp(1)$,
\item normal subgroups $A_0 \lhd A$ and $B_0 \lhd B$ such that $A/A_0 \cong B/B_0 \cong F$, where $F$ is some abstract finite group, and
\item an automorphism of $F$.
\end{enumerate}

As discussed in Section \ref{sec:statement-problem}, we are interested in the subgroup $\Gamma$ up to conjugation in the normaliser of the $\Sp(1) \times \Sp(1)$ in $\Spin(8)$.  We can use the freedom to conjugate in $\Sp(1) \times \Sp(1)$ to choose $A$ and $B$ to be fixed ADE subgroups of $\Sp(1)$, thus given by the groups in Table \ref{tab:ADE}.  Furthermore we can use the freedom to swap the two $\Sp(1)$ subgroups, which is induced by conjugation in $\Spin(8)$, in order to order the two ADE subgroups, say, alphabetically.  This still leaves the freedom to conjugate by their normalisers $N(A)$ and $N(B)$ in $\Sp(1)$.

Next we have to choose normal subgroups $A_0 \lhd A$ and $B_0 \lhd B$ with isomorphic factor groups $A/A_0 \cong B/B_0$.  Any pair $A$, $B$ possesses such normal subgroups, since we can always take $A_0=A$ and $B_0=B$, but some pairs will also have more interesting normal subgroups as we will see in the next section.  Conjugation by the normaliser $N(A)$ permutes the normal subgroups of $A$ (and similarly for $B$) and two normal subgroups so related are deemed to be equivalent for our purposes.  This means that once we choose $A_0$ and $B_0$ in their equivalence classes, the only ingredient left to choose is the automorphism $\tau$ of $F$ and the only freedom left is conjugation by $N_0(A) \times N_0(B)$, where $N_0(A)$ and $N_0(B)$ are the stabilisers of $A_0$ in $N(A)$ and of $B_0$ in $N(B)$.

As we now show, the group $N_0(A) \times N_0(B)$ acts on $\Aut(F)$.  Let us denote by $a \mapsto [a]$ the composition $A \to A/A_0 \to F$ (where the second map is $i_A$) and similarly by $b \mapsto [b]$ the composition $B \to B/B_0 \to F$, where the second map is $i_B$.  Let $\mu \in \Aut(F)$ be the automorphism of $F$ induced by conjugation with $x \in N_0(A)$.  In other words, $\mu[a] = [x a x^{-1}]$.  One can check that this is well-defined and defines an automorphism of $F$.  Similarly, let $\nu$ be the automorphism of $F$ induced by conjugation with $y \in N_0(B)$.  Then it is easy to verify that $(a,b) \in A \times_{(F,\tau)} B$ if and only if $(x a x^{-1}, y b y^{-1}) \in A \times_{(F,\nu \circ \tau \circ \mu^{-1})} B$.
In other words, $(x,y) \in N_0(A) \times N_0(B)$ acts on $\tau \in \Aut(F)$ by $\tau \mapsto \nu \circ \tau \circ \mu^{-1}$ and two automorphisms of $F$ which are so related give rise to equivalent orbifolds.  This suggests that we introduce the set of \emph{twists}
\begin{equation}
  \label{eq:twists}
  \Twist(F) := \Aut(F)/\left(N_0(A)\times N_0(B)\right)~,
\end{equation}
which is in general \emph{not} a group, but merely the quotient of $\Aut(F)$ by the above action of $N_0(A)\times N_0(B)$.  In practice we will make a choice of representative automorphism for each equivalence class in $\Twist(F)$.  The notation is also not particularly good, since $\Twist(F)$ does not just depend on the abstract group $F$ but on $F$ as a common quotient of $A$ and $B$.

Notice that because $A_0$ and $B_0$ are normal subgroups, $N_0(A)$ and $N_0(B)$ contain the inner automorphisms of $A$ and $B$ and these surject onto the inner automorphisms of $F$, whence in order to compute the set of twists we can in the first instance restrict to the group $\Out(F)$ of outer automorphisms of $F$ and then investigate which outer automorphisms of $F$ are related by the action of $N_0(A)\times N_0(B)$.

In summary, $\Gamma$ is determined up to isomorphisms, by the following data:
\begin{enumerate}
\item an unordered pair $\{A,B\}$ of subgroups of $\Sp(1)$ taken from Table~\ref{tab:ADE};
\item normal subgroups $A_0 \lhd A$ and $B_0 \lhd B$, each representing its equivalence class under conjugation by the normalisers in $\Sp(1)$ of $A$ and $B$, respectively, and such that $A/A_0 \cong B/B_0 \cong F$ where $F$ is some abstract finite group, and
\item a representative automorphism $\tau \in \Twist(F)$.
\end{enumerate}

Finally, we make some comments about the smooth case.  As shown in \cite{deMedeiros:2009pp} those subgroups of $\Sp(1) \times \Sp(1)$ which act freely on $S^7$ are such that $A_0 = B_0 = \left\{1\right\}$, whence $A \cong B$.  They are given by the graph in $A \times A$ of an automorphism $\tau : A \to A$ of a fixed ADE subgroup $A < \Sp(1)$.

We end with the observation that quotients with $\eN>4$ are automatically smooth: indeed, if $S^7/\Gamma$ has $\eN\geq 5$, then $\Gamma < \Spin(3)$, where the $\Spin(3)$ subgroup is the diagonal $\Sp(1)$ subgroup of $\Sp(1) \times \Sp(1)$.  So the elements of $\Gamma$ are of the form $(a,a)$ for $a \in A$ some ADE subgroup of $\Sp(1)$.  Therefore we see from \eqref{eq:normal} that $A_0 = B_0 = \{1\}$.  Hence these subgroups are already considered in \cite{deMedeiros:2009pp}, whence the novelty in this paper consists of extending the classification of $\eN=4$ quotients by the inclusion of singular quotients.  Since homogeneity implies smoothness, the fact that $\eN>4$ orbifolds are actually smooth is of course consistent with the homogeneity conjecture \cite{FMPHom,EHJGMHom,JMF-HC-Lecs} which says, in this context, that backgrounds preserving mode than half the supersymmetry are homogeneous.

\section{Quotients of ADE subgroups}
\label{sec:quotients-ade}

In this section we record the normal subgroups of the ADE subgroups and the corresponding factor groups.  These are the abstract finite groups $F$ appearing in Goursat's Lemma.  The results are summarised in Table~\ref{tab:normal}.

\subsection{Cyclic groups}
\label{sec:cyclic-groups}

Every subgroup of an abelian group is normal and every subgroup of a cyclic subgroup is cyclic.  Indeed, let $G = \left<t\right>$ with $t^n = 1$  be a cyclic group of order $n$ and let $N\lhd G$ be a normal subgroup.  Let $0<k<n$ be the smallest integer such that $t^k \in N$.  Then we claim that $N = \left<t^k\right>$.  Indeed, let $t^\ell \in N$ and write $\ell = qk + r$, where the remainder $0\leq r < k$.  Since $t^\ell$ and $(t^k)^q$ belong to $N$, so does $t^r$ and this would violate the minimality of $k$ unless $r=0$.

The ADE group $\sA_{n-1}$ is cyclic of order $n$.  Its normal groups are $\sA_{m-1}$ for any $m$ which is a divisor of $n$ and the factor groups are thus also cyclic groups isomorphic to $\ZZ_k$ with $km=n$.

\subsection{Binary dihedral groups}
\label{sec:binary-dihedr-groups}

The normal subgroups of the binary dihedral groups can be determined in a similar way to how those of the dihedral groups are obtained (see, e.g., \cite[Lemma~1.1]{MR0308206} for $D_{4k+2}$ and \cite[Lemma~2.1]{MR0330236} for $D_{4k}$).

Let $\sD_{n+2}$ be the subgroup of $\Sp(1)$ generated by the quaternions $s=j$ and $t = e^{i\pi/n}$.  Abstractly, it has the presentation
\begin{equation}
 \sD_{n+2} = \left<s,t \middle | s^2 = t^n = (st)^2\right>~.
\end{equation}
Notice that $s^2= t^n = (st)^2 = -1$.

We first observe that any subgroup of the cyclic subgroup generated by $t$ is normal.  To see this simply notice that from $stst=-1$ one finds
\begin{equation}
  sts^{-1} = -sts = -ststt^{-1} = t^{-1} \implies st^k s^{-1} = t^{-k}~,
\end{equation}
but if $t^k$ belongs to a subgroup, then so does its inverse $t^{-k}$.  Since $t$ has order $2n$, any subgroup is generated by $t^k$ for some $k|2n$.  We must distinguish between two cases: when $k|n$ and $k\not|\,n$, this latter case forcing $k$ to be even.

If $k|n$, then let $n=kl$.  Then the normal subgroup $\left<t^k\right> \cong \ZZ_{2l}$ is the ADE subgroup $\sA_{2l-1}$.  Let $F$ denote the factor group $\sD_{n+2}/\sA_{2l-1}$ and let $x\mapsto [x]$ denote the canonical surjection $\sD_{n+2} \to F$.  Then since $-1\in \sA_{2l-1}$, $[s]^2 = 1$, $([s][t])^2 = 1$ and $[t]^k = 1$, whence $F \cong D_{2k}$, with the understanding that $D_2 \cong\ZZ_2$ and $D_4 \cong \ZZ_2 \times \ZZ_2$.  In other words, $\sD_{n+2}$ for any $n$, has a normal cyclic subgroup of index 2.  We will see below that for $n$ even, there are in addition two nonabelian normal subgroups of index 2.

If $k\not|\,n$, then since $k|2n$, we see that $k$ must be even: say, $k=2p$ with $p|n$.  Moreover, $n/p$ must be odd, otherwise $k|n$.  In summary, $n = p(2l+1)$.  The normal subgroup $\left<t^k\right> \cong \ZZ_{2l+1}$ is the ADE subgroup $\sA_{2l}$ and let $F$ now denote the factor group $\sD_{n+2}/\sA_{2l}$.  Since $-1 \not\in \sA_{2l}$, $[s]^2=[-1]$ and $([s][t])^2=[-1]$.  Moreover, one has that
\begin{equation}
  [-1] = [t]^n = [t]^{p(2l+1)} = [t]^{kl} [t]^p = [t]^p~,
\end{equation}
whence $F \cong 2D_{2p} = 2D_k$, with the understanding that $2D_2 \cong \ZZ_4$.

Any other normal subgroup $N$ which is not a subgroup of $\left<t\right>$ must have an element of the form $st^k$ for some $k$.  Since
\begin{equation}
  (st^k)^2 = st^kst^k = -s t^k s^{-1} t^k = - t^{-k} t^k = -1~,
\end{equation}
we see that $-1\in N$.  Since $s^2 = stst$, we see that $tst=s$, whence conjugating $st^k$ by $t$, we find
\begin{equation}
  tst^kt^{-1} = tst t^{k-2} = s t^{k-2}~.
\end{equation}
Therefore if $s\in N$, $N$ contains all $st^{\mathrm{even}}$ ,whereas if $st \in N$, $N$ contains all $st^{\mathrm{odd}}$.  Similarly,
\begin{equation}
  st^k st^l = -t^{-k}t^l = - t^{l -k }~,
\end{equation}
whence $N$ contains all the $t^{\mathrm{even}}$.  If $n=2p+1$ is odd, so that $t^{2p+1} = -1$, $t = -t^{-2p}$ also belongs to $N$ and hence $N=\sD_{n+2}$ is not a proper subgroup.  If $n=2p$ is even, then we have two index-2 normal subgroups: the normal subgroup generated by $s$ and the normal subgroup generated by $st$.  Both of these normal subgroups are related by an outer automorphism which sends $(s,t)$ to $(st,t)$.  In terms of quaternions, it is conjugation by $e^{i\pi/2n}$.  In summary, we have that $\sD_{2p+2}$ has (up to equivalence) two proper normal subgroups of index 2: $\sD_{p+2}=\left<s,t^2\right>$ and $\sA_{4p-1} = \left<t\right>$.

\subsection{Binary tetrahedral group}
\label{sec:binary-tetr-group}

It is easy to determine the normal subgroups of $\sE_6$ from the knowledge of the conjugacy classes.  These are tabulated in Table~\ref{tab:classes2T}, which is borrowed from \cite{deMedeiros:2009pp}.   In the table, $s=\frac{(1+i)(1+j)}{2}$ and $t=\frac{(1+j)(1+i)}{2}$.

\begin{table}[h!]
  \centering
  \caption{Conjugacy classes of $\sE_6$}
  \begin{tabular}{*{8}{>{$}c<{$}|}}
    \text{Class} & 1 & -1 & s & t & t^2 & s^2 & st \\
    \text{Size} & 1 & 1 & 4 & 4 & 4 & 4 & 6 \\
    \text{Order} & 1 & 2 & 6 & 6 & 3 & 3 & 4
  \end{tabular}
  \label{tab:classes2T}
\end{table}

A normal subgroup is made out of conjugacy classes and its order must divide the order of the group, in this case $24$.  There are 8 divisors of $24$: $1,2,3,4,6,8,12,24$, the first and last correspond to the improper normal subgroups.  Every subgroup contains the identity element, so the class $\{1\}$ has to be present.  There is precisely one union of conjugacy classes of order $2$, namely the centre $\sA_1 = \{\pm 1\}$ which is clearly normal.  There is no way to get a normal subgroup of orders $3$ or $4$.   This means that all proper normal subgroups must have even order and hence, in particular, they all contain $-1$.  Now, $s$ and $-s$ are in different conjugacy classes, since $-s$ is conjugate to $t^2$ and similarly, $t$ and $-t$ are in different conjugacy classes, since $-t$ is conjugate to $s^2$.  So if a normal subgroup contains any of the classes of size $4$, by taking products, one sees it must contain at least three such classes and hence, by order, it can only be the full group.  Therefore the only other proper normal subgroup could be the one consisting of the centre and the conjugacy class of size $6$.  One can check that this is indeed a normal subgroup of index 3.  It consists of the quaternion units $\{\pm 1, \pm i, \pm j, \pm k\}$, which is the ADE subgroup $\sD_4$.  In summary, there are two proper normal subgroups of $\sE_6$: $\sA_1$ and $\sD_4$, whence their factor groups are isomorphic to $T$ and $\ZZ_3$, respectively.

\subsection{Binary octahedral group}
\label{sec:binary-octah-group}

Urged on by our success with the binary tetrahedral group, we subject the binary octahedral group to a similar analysis.  The conjugacy classes are now tabulated in Table~\ref{tab:classes2O}.    In the table, $s=\frac{(1+i)(1+j)}{2}$ and $t=\frac{1+i}{\sqrt{2}}$.

\begin{table}[h!]
  \centering
  \caption{Conjugacy classes of $\sE_7$}
  \begin{tabular}{*{9}{>{$}c<{$}|}}
    \text{Class} & 1 & -1 & s & t & s^2 & t^2 & t^3 & st \\
    \text{Size} & 1 & 1 & 8 & 6 & 8 & 6 & 6 & 12\\
    \text{Order} & 1 & 2 & 6 & 8 & 3 & 4 & 8 & 4
  \end{tabular}
  \label{tab:classes2O}
\end{table}

Let $N \lhd \sE_7$ be a proper normal subgroup.  Then it must have order equal to a proper divisor of $48$ (the order of $\sE_7$), which is one of $2,3,4,6,8,12,16,24$.  Since it is composed of conjugacy classes and the class of the identity must be included, this leaves the following possibilities: $2,8,16,24$.  The subgroup of order $2$ is again the centre $\sA_1$, whereas the normal subgroups of order $8$ must have class equation (in $\sE_7$) $1+1+6$.  The conjugacy classes of size $6$ are those of $t$, $t^2$ and $t^3$.  Clearly, if $t \in N$, then all three classes must arise.  This already means that $N$ must be the whole group.  Similarly if $t^3\in N$ arises, then also $t=-t^3\in N$, and again $N$ is not proper.  So this leaves only the class of $t^2$.  This does form a normal subgroup of index 6: namely, $\sD_4$ consisting of the quaternion units $\{\pm 1, \pm i, \pm j, \pm k\}$.  No normal subgroup of order $16$ exists, since its class equation is $1+1+6+8$, but it is easy to show that if one of the size-8 conjugacy classes belongs to $N$, then so must be other.  This is because one is the class of $s$ and the other is the class of $-s=s^2$.  This then leaves the possibility of a normal subgroup of index 2, whose class equation is $1+1+6+8+8$.  One can check that this is the ADE subgroup $\sE_6$.  In summary, the proper normal subgroups of $\sE_7$ are $\sA_1$, $\sD_4$ and $\sE_6$.  The corresponding factor groups are isomorphic to $O$, $D_6$ and $\ZZ_2$, respectively.

\subsection{Binary icosahedral group}
\label{sec:binary-icos-group}

The binary icosahedral case is simpler.  We recall that it has the abstract presentation
\begin{equation}
  \label{eq:2I}
  \sE_8 = \left<s,t \middle | s^3 = t^5= (st)^2\right>~,
\end{equation}
where $s^3= t^5 = (st)^2 = -1$ and $s=\frac{(1+i)(1+j)}{2}$ and $t=\frac{\phi + \phi^{-1}i + j}{2}$, with $\phi=\half(1+\sqrt{5})$ is the Golden Ratio.  First we can argue as follows.  The factor group of $\sE_8$ by the centre is the icosahedral group, which is a simple group isomorphic to $A_5$.  Therefore it has no proper normal subgroups.  Hence if $N\lhd \sE_8$ is a proper normal subgroup, its projection to $I$ must either be the whole group or else the identity.  The latter corresponds to $N$ being the centre $\sA_1$, whereas the former situation cannot arise.  Indeed, if there were a subgroup of $\sE_8$ isomorphic to $I$ under the projection, it could not contain $-1$, but then it must either contain $st$ or $-st$ both of which square to $-1$.

\begin{table}[h!]
  \caption{Conjugacy classes of $\sE_8$}
  \centering
  \begin{tabular}{*{10}{>{$}c<{$}|}}
    \text{Class} & 1 & -1 & t & t^2 & t^3 & t^4 & s & s^4 & st \\
    \text{Size} & 1 & 1 & 12 & 12 & 12 & 12 & 20 & 20 & 30\\
    \text{Order} & 1 & 2 & 10 & 5 & 10 & 5 & 6 & 3 & 4
  \end{tabular}
 \label{tab:classes2I}
\end{table}

Alternatively, one can examine the possible conjugacy classes, tabulated in Table~\ref{tab:classes2I}, and check that the only proper divisor of $120=|\sE_8|$ which can be obtained by adding sizes of conjugacy classes including the class of the identity, is $2$.

In summary, the only proper normal subgroup of $\sE_8$ is $\sA_1$, with factor group isomorphic to $I$.

\subsection{Summary}
\label{sec:summary}

Table~\ref{tab:normal} summarises the proper normal subgroups of the ADE subgroups of $\Sp(1)$ up to conjugation in $\Sp(1)$ and including the isomorphism type of the corresponding factor group.  Of course, to this table one should add the improper normal subgroups: namely, $\{\one\}$ and the whole group itself.  It is important to stress the fact that whereas $N$ and $G$ are subgroups of $\Sp(1)$, $F$ is an abstract group obtained as a factor group.  It is for this reason that even if $F$ happens to be abstractly isomorphic to one of the ADE subgroups of $\Sp(1)$ we do not use the notation $\sA$, $\sD$, $\sE$, which we reserve for the finite subgroups of $\Sp(1)$ in Table~\ref{tab:ADE}.  We remark that in either of the two cases $\sA_{2k-1} \lhd \sD_{kl+2}$ and $\sA_{2k} \lhd \sD_{l(2k+1)+2}$, the integer $l$ can take the value $1$, in which cases $D_2 \cong \ZZ_2$ and $2D_2 \cong \ZZ_4$.

\begin{table}[h!]
  \caption{Normal (proper) subgroups of ADE subgroups}
  \centering
  \begin{tabular}[t]{>{$}l<{$}|>{$}l<{$}}
    N\lhd G & G/N\\\hline
    \sA_{k-1} \lhd \sA_{kl-1} & \ZZ_l\\
    \sA_{2k-1} \lhd \sD_{kl+2} & D_{2l}\\
    \sA_{2k} \lhd \sD_{l(2k+1)+2} & 2D_{2l}\\
    \sD_{k+2} \lhd \sD_{2k+2} & \ZZ_2\\
    \sA_1 \lhd \sE_6 & T
\end{tabular}
 \qquad\qquad
  \begin{tabular}[t]{>{$}l<{$}|>{$}l<{$}}
    N\lhd G & G/N\\\hline
   \sD_4 \lhd \sE_6 & \ZZ_3\\
    \sA_1 \lhd \sE_7 & O\\
    \sD_4 \lhd \sE_7 & D_6\\
    \sE_6 \lhd \sE_7 & \ZZ_2\\
    \sA_1 \lhd \sE_8 & I
 \end{tabular}
 \label{tab:normal}
\end{table}

\section{Subgroups of $\Spin(4)$}
\label{sec:subgroups-spin4}

In this section we list the finite subgroups of $\Spin(4)$ (up to automorphisms of $\Spin(4)$) obtained via Goursat's theory.  Recall that the data determining such a subgroup is a pair $A$, $B$ of ADE subgroups of $\Sp(1)$, normal subgroups $A_0 \lhd A$ and $B_0\lhd B$ with isomorphic factor groups $A/A_0 \cong B/B_0 \cong F$ and an automorphism $\tau$ representing a class in $\Twist(F)$.  Therefore we need to determine which pairs of ADE subgroups have isomorphic factor groups and the group of outer automorphisms of all possible factor groups.

\subsection{Compatible ADE subgroups}
\label{sec:comp-ade-subgr}

Let us say that two ADE subgroups $A$ and $B$ are \emph{compatible} if they admit isomorphic factor groups.   Compatibility is clearly an equivalence relation, and we can read off the equivalence classes, indexed by the factor group $F$, from Table~\ref{tab:normal}.  This is displayed in Table~\ref{tab:compatible}.  In the line corresponding to the factor group $\ZZ_2$, the notation is such that $\sD_{k+2}$ has normal subgroup $\sA_{2k-1}$, whereas $\sD'_{2k+2}$ is $\sD_{2k+2}$ but with normal subgroup $\sD_{k+2}$.  Except for this ambiguity, the factor group determines (the equivalence class of) the normal subgroup uniquely.

\begin{table}[h!]
  \caption{Compatible ADE subgroups}
  \centering
  \begin{tabular}[t]{>{$}l<{$}|>{$}l<{$}}
    F & G\\\hline
    \{1\} & \text{all}\\
    \ZZ_2 & \sA_{2k-1}, \sD'_{2k+2}, \sD_{k+2},\sE_7 \\
    \ZZ_3 & \sA_{3k-1}, \sE_6\\
    \ZZ_4 & \sA_{4k-1}, \sD_{(2k+1)+2}\\
    \ZZ_{l\geq 5} & \sA_{kl -1} \\
    D_6 & \sD_{3k+2}, \sE_7\\
    D_{2l\neq 2,6} & \sD_{kl+2}\\
 \end{tabular}
  \qquad\qquad
  \begin{tabular}[t]{>{$}l<{$}|>{$}l<{$}}
    F & G\\\hline
    2D_{2l} & \sD_{l(2k+1)+2}\\
    T & \sE_6\\
    O & \sE_7\\
    I & \sE_8\\
    2T & \sE_6\\
    2O & \sE_7\\
    2I & \sE_8
  \end{tabular}
  \label{tab:compatible}
\end{table}

\subsection{Outer automorphisms of some finite groups}
\label{sec:outer-autom-some}

It remains to determine the group of outer automorphisms of the factor groups.

\subsubsection{Outer automorphisms of $\ZZ_n$}
\label{sec:outer-autom-zz_n}

Since $\ZZ_n$ is abelian, every automorphism is outer.  Let $x$ be a generator of $\ZZ_n$, so that $x$ has order $n$.  Then under an automorphism, $x \mapsto x^r$, where $x^r$ too has order $n$, whence $r$ is coprime to $n$.  In other words, $r$ defines an element in the multiplicative group $\ZZ_n^\times$ of units in $\ZZ_n$.  This group has order $\phi(n)$, the value of Euler's totient function.

\subsubsection{Outer automorphisms of $D_{2n}$}
\label{sec:outer-autom-d_2n}

Let $D_{2n}$ denote the group with presentation and enumeration
\begin{equation}
  \label{eq:dihedral}
  D_{2n} = \left<x,y \middle| x^2 = y^n = (x y)^2 = 1\right> = \left\{y^p\middle | 0\leq p < n\right\} \cup \left\{x y^p\middle | 0\leq p < n\right\}~.
\end{equation}
For $n> 2$, the automorphism group $\Aut(D_{2n})$ of $D_{2n}$ is the affine group on the rotational subgroup $\ZZ_n$: namely, $\ZZ_n^\times \ltimes \ZZ_n$.  Its action on the group elements is the following.  Let $(a,b) \in \ZZ_n^\times \ltimes \ZZ_n$ and let $\tau_{a,b}$ denote the corresponding automorphism.  Then
\begin{equation}
  \tau_{a,b} (y^p) = y^{ap}  \qquad\text{and}\qquad  \tau_{a,b}(x y^p) = x y^{ap + b}~.
\end{equation}
On the other hand, the group of inner automorphisms is the subgroup of $\Aut(D_{2n})$ generated by conjugation by the generators: $x$ and $y$:
\begin{equation}
  \begin{aligned}[m]
    x y^p x^{-1} &= y^{-p}\\
    x^2 y^p x^{-1} & =  x y^{-p}
  \end{aligned}
\qquad\qquad
  \begin{aligned}[m]
    y y^p y^{-1} &= y^p\\
    y x y^p y^{-1} & = x y^{p-2}~.
  \end{aligned}
\end{equation}
In other words, the subgroup of inner automorphisms is generated by $\tau_{-1,0}$ and $\tau_{1,2}$.  We must distinguish between $n$ even or odd.  If $n$ is odd, then all the translations in $\Aut(D_{2n})$ are inner, whereas if $n$ is even only the even translations are inner.  Thus, if $n$ is odd, the subgroup of inner automorphisms is isomorphic to $\ZZ_2 \ltimes \ZZ_n$ and hence the group $\Out(D_{2n})$ of outer automorphisms is isomorphic to the factor group $\ZZ_n^\times/\left<-1\right>$ of $\ZZ_{n}^\times$ by the order-2 subgroup generated by $-1$.  On the other hand, if $n$ is even then the subgroup of inner automorphisms is isomorphic to $\ZZ_2 \ltimes \ZZ_{n/2}$, whence $\Out(D_{2n})$ is isomorphic to the direct product of the factor group $\ZZ_n^\times/\left<-1\right>$ and the order-2 group of translations modulo even translations.

If $n=2$, then $D_4 \cong \ZZ_2 \times \ZZ_2$ is  the Klein \emph{Viergruppe}, whose automorphism group permutes all the elements of order $2$, whence $\Aut(D_4) \cong S_3 \cong D_6$, and since $D_4$ is abelian, $\Out(D_4) \cong D_6$, as well.

\subsubsection{Outer automorphisms of $2D_{2n}$}
\label{sec:outer-autom-2d_2n}

This is very similar to the previous case.  The group $2D_{2n}$ now admits the presentation and enumeration
\begin{equation}
  2D_{2n} = \left<s,t \middle| s^2 = t^n = (st)^2 \right> = \left\{t^p\middle | 0\leq p < 2n\right\} \cup \left\{st^p\middle | 0\leq p < 2n\right\}~.
\end{equation}
For $n>2$, the group $\Aut(2D_{2n})$ of automorphisms is again an affine group $\ZZ_{2n}^\times \ltimes \ZZ_{2n}$ with $(a,b) \in \ZZ_{2n}^\times \ltimes \ZZ_{2n}$ acting via $\tau_{a,b}$ defined by
\begin{equation}
  \tau_{a,b} (t^p) = t^{ap}  \qquad\text{and}\qquad  \tau_{a,b}(s t^p) = s t^{ap + b}~.
\end{equation}
The subgroup of inner automorphisms is generated by $\tau_{-1,0}$ and $\tau_{1,2}$, and is thus isomorphic to $\ZZ_2 \ltimes \ZZ_n$.  The group $\Out(2D_{2n})$ of outer automorphisms is the direct product $\ZZ_{2n}^\times/\left<-1\right> \times \ZZ_2$ where the $\ZZ_2$ factor is the order-2 group of translations modulo even translations.

Again, we need to distinguish the case $n=2$.  In this case, the automorphism group of $2D_4$ is isomorphic to the octahedral group $O$ and the group of outer automorphisms is isomorphic to $D_6$.  We can see how to describe the automorphisms in a very concrete way by using the ADE subgroups of $\Sp(1)$.  Indeed, $2D_4$ is isomorphic to the ADE subgroup $\sD_4$ and this sits as a normal subgroup of $\sE_7$.  The action of $\sE_7$ on $\sD_4$ is via automorphisms.  The centre acts trivially, so the image of $\sE_7$ in $\Aut(2D_4)$ is isomorphic to $\sE_7/\sA_1 \cong O$.  The group of outer automorphisms is isomorphic to the quotient $\sE_7/\sD_4 \cong D_6$.

Finally, we consider the case $n=1$.  Now $2D_2 \cong \ZZ_4$ and the automorphism group is isomorphic to $\ZZ_4^\times \cong \ZZ_2$.

\subsubsection{Outer automorphisms of $T$}
\label{sec:outer-autom-t}

The tetrahedral group $T$ is isomorphic to the alternating group $A_4$.  The automorphism group is $S_4$ and since $A_4$ has no centre, the group of inner automorphisms is isomorphic to $A_4$, whence the outer automorphism group is isomorphic to $\ZZ_2$.  In terms of the presentation
\begin{equation}
  T = \left<x,y | x^3 = y^3 = (x y)^2 = 1\right>~,
\end{equation}
a representative for the generator of the group of outer automorphisms is the automorphism which swaps $x$ and $y$.

\subsubsection{Outer automorphisms of $O$}
\label{sec:outer-autom-o}

The octahedral group is isomorphic to the symmetric group $S_4$, whence its automorphism group is $S_4$.  Since the centre is trivial, the group of inner automorphisms is also isomorphic to $S_4$, whence the group of outer automorphisms is trivial.

\subsubsection{Outer automorphisms of $I$}
\label{sec:outer-autom-i}

The icosahedral group is isomorphic to the alternating group $A_5$, which is a simple group.   The automorphism group is $S_5$ and the subgroup of inner automorphisms is isomorphic to $A_5$, whence the group of outer automorphisms is isomorphic to $\ZZ_2$.  In terms of 
permutations, $A_5$ is generated by $s=(142)$ and $t=(12345)$ and the outer automorphisms are generated by conjugation by an odd permutation, e.g., $(35)$, which leaves $s$ invariant and sends $t$ to $(12543)$.

\subsubsection{Outer automorphisms of $2T$}
\label{sec:outer-autom-2t}

This case was treated, for example, in \cite[§7.3]{deMedeiros:2009pp}.  The group of outer automorphisms isomorphic to $\ZZ_2$, and it is generated by the automorphism which exchanges the generators $\frac{(1+i)(1+j)}{2}$ and $\frac{(1+j)(1+i)}{2}$ in Table~\ref{tab:ADE}.  However this automorphism is obtained by conjugation with $\frac{i+j}{\sqrt{2}}$ in $\Sp(1)$.  As shown in \cite{deMedeiros:2009pp}, this means that the quotient involving the nontrivial outer automorphism is equivalent to the one involving the identity automorphism.

\subsubsection{Outer automorphisms of $2O$}
\label{sec:outer-autom-2o}

This case was treated, for example, in \cite[§7.4]{deMedeiros:2009pp}.  The automorphism group $\Aut(2O)\cong O \times \ZZ_2$, where $O$ corresponds to the inner automorphisms, whence $\Out(2O) \cong \ZZ_2$, whose generator is represented by the automorphism which fixes the first generator $\frac{(1+i)(1+j)}{2}$ in Table~\ref{tab:ADE}, and changes the sign of the second generator $\frac{1+i}{\sqrt{2}}$.

\subsubsection{Outer automorphisms of $2I$}
\label{sec:outer-autom-2i}

Finally, this was treated in \cite[§7.5]{deMedeiros:2009pp}.  The group of outer automorphisms is again $\ZZ_2$, whose generator is represented by the automorphism which leaves the first generator $\frac{(1+i)(1+j)}{2}$ in Table~\ref{tab:ADE} alone and sends the second generator $\frac{\phi + \phi^{-1}i + j}{2}$ to $\frac{-\phi^{-1} - \phi i + k}{2}$.

Table~\ref{tab:outer} summarises the above considerations.

\begin{table}[h!]
  \caption{Outer automorphisms of factor groups}
  \centering
  \begin{tabular}[t]{>{$}l<{$}|>{$}l<{$}}
    F & \Out(F)\\\hline
    \ZZ_l & \ZZ_l^\times\\
    D_{4l>4} & \ZZ_{2l}^\times/\left<-1\right> \times \ZZ_2\\
    D_{4l + 2} & \ZZ_{2l+1}^\times/\left<-1\right>\\
    D_4 & D_6\\
    2D_{2l>4} &  \ZZ_{2l}^\times/\left<-1\right> \times \ZZ_2\\
    2D_{4} &  D_6
\end{tabular}
  \qquad\qquad
  \begin{tabular}[t]{>{$}l<{$}|>{$}l<{$}}
    F & \Out(F)\\\hline
    T & \ZZ_2\\
    O & \left\{1\right\}\\
    I & \ZZ_2\\
    2T & \ZZ_2\\
    2O & \ZZ_2\\
    2I & \ZZ_2
  \end{tabular}
  \label{tab:outer}
\end{table}

\subsection{The finite subgroups of $\Spin(4)$}
\label{sec:finite-subgr-spin4}

Let $A$ and $B$ be ADE subgroups of $\Sp(1)$.  Then as described above, all finite subgroups of $\Spin(4)$ are fibred products $A\times_{(F,\tau)} B$, where $A,B$ admit factor groups isomorphic to $F$ and $\tau \in \Out(F)$ is an outer automorphism of $F$.

As a special case of this construction we have those subgroups where $F=\{1\}$ (whence $\tau = 1$), which correspond to the direct product $A \times B$ and tabulated in Table~\ref{tab:products}, where we allow $\sA_n$ to include the trivial group $\sA_0$ as a special case.  All these give rise to orbifolds with $\eN=4$.

\begin{table}[h!]
  \caption{Product subgroups of $\Spin(4)$}
  \centering
  \begin{tabular}[t]{>{$}l<{$}|>{$}l<{$}}
    \multicolumn{1}{c}{$\Gamma$}& \text{Order}\\\hline
    \sA_{n-1} \times \sA_{m-1} & nm\\
    \sA_{n-1} \times \sD_{m+2} & 4nm\\
    \sA_{n-1} \times \sE_6 & 24n\\
    \sA_{n-1} \times \sE_7 & 48n\\
    \sA_{n-1} \times \sE_8 & 120n\\
    \sD_{n+2} \times \sD_{m+2} & 16nm\\
    \sD_{n+2} \times \sE_6 & 96n\\
    \sD_{n+2} \times \sE_7 & 192n
  \end{tabular}
  \qquad\qquad
  \begin{tabular}[t]{>{$}l<{$}|>{$}l<{$}}
    \multicolumn{1}{c}{$\Gamma$} & \text{Order}\\\hline
    \sD_{n+2} \times \sE_8 & 480n\\
    \sE_6 \times \sE_6 & 576\\
    \sE_6 \times \sE_7 & 1152 \\
    \sE_6 \times \sE_8 & 2880\\
    \sE_7 \times \sE_7 & 2304\\
    \sE_7 \times \sE_8 & 5760\\
    \sE_8 \times \sE_8 & 14400
  \end{tabular}
  \label{tab:products}
\end{table}

Another special case of this construction is where $A=B=F$.  Such groups are abstractly isomorphic to $A$, but the embedding in $\Spin(4)$ is as the graph of an (outer) automorphism.  These are precisely the subgroups leading to the smooth quotients classified in \cite{deMedeiros:2009pp} and tabulated in Table~\ref{tab:smooth}, which is borrowed from \cite{deMedeiros:2009pp}.  The cases $\sA_{\leq 3}$, $\sA_5$, $\sD_4$, $\sD_5$ and $\sE_6$ have no $\eN=4$ quotients because they either have no nontrivial outer automorphisms or else the nontrivial outer automorphisms give rise to equivalent quotients to the case of trivial outer automorphisms.   The automorphisms $\mu$ and $\nu$ in the $\sE_7$ and $\sE_8$ cases represent the unique nontrivial outer automorphisms of those groups.  If no automorphism is shown we take the diagonal subgroup consisting of the graph of the identity.

\begin{table}[h!]
  \caption{Smooth quotients with $\eN\geq 4$}
  \centering
  \begin{tabular}{>{$}l<{$}|>{$}l<{$}|>{$}l<{$}|>{$}c<{$}}
    A = B = F & \tau \in \Twist(A) & |A \times_{(A,\tau)} A| & \eN \\\hline
    \sA_1 & 1 & 2 & 8\\
    \sA_{n-1\geq 2} & 1 & n & 6 \\
    \sD_{n+2\geq 4} & 1 & 4n & 5 \\
    \sE_6 & 1 & 24 & 5 \\
    \sE_7 & 1 & 48 & 5 \\
    \sE_8 & 1 & 120 & 5 \\
    \sA_{n-1\geq 4} & 1\neq r \in \ZZ^\times_n/\left<-1\right> & n& 4 \\
    \sD_{n+2\geq6} & 1\neq r \in \ZZ^\times_{2n}/\left<-1\right> & 4n & 4 \\
    \sE_7 & \mu & 48 & 4\\
    \sE_8 & \nu & 120 & 4
  \end{tabular}
  \label{tab:smooth}
\end{table}

Finally, we have the rest of the subgroups $A\times_{(F,\tau)} B$, all of which give rise to $\eN=4$ orbifolds.  We list them in Table~\ref{tab:remaining}, which contains the pair $A,B$ of ADE subgroups of $\Sp(1)$, their common factor group $F$, the set $\Twist(F)$ of inequivalent twists, and the order of $A\times_{(F,\tau)} B$, which is given by $|A||B|/|F|$ in all cases.  We recall that the set $\Twist(F)$ of inequivalent twists is defined by equation \eqref{eq:twists} and is determined in the next section.  The notation $\{1\}$ simply means that there is a unique equivalence class of possible twists and we may and will choose the identity automorphism as its representative.  Let us simply advance that in both appearances of $D_6/{\sim}$, this set is ere are only two inequivalent twists.

\begin{table}[h!]
  \caption{Remaining finite subgroups of $\Spin(4)$}
  \centering
  \begin{tabular}[t]{*{4}{>{$}l<{$}|}>{$}l<{$}}
   A & B & F & \Twist(F) & |A\times_{(F,\tau)}B| \\\hline
   \sA_{kl-1} & \sA_{ml-1} & \ZZ_l & \ZZ_l^\times/\left<-1\right> & kl m\\
   \sA_{2k-1} & \sD'_{2m+2} & \ZZ_2 & \{1\} & 8km\\
   \sA_{2k-1} & \sD_{m+2} & \ZZ_2 & \{1\} & 4km\\
   \sA_{4k-1} & \sD_{2m+3} & \ZZ_4 & \{1\} &  4k(2m+1)\\
   \sA_{3k-1} & \sE_6 & \ZZ_3 & \{1\} & 24k\\
   \sA_{2k-1} & \sE_7 & \ZZ_2 & \{1\} & 48k\\
   \sD'_{2k+2} & \sD'_{2m+2} & \ZZ_2 & \{1\} & 32km\\
   \sD_{k+2} & \sD_{m+2} & \ZZ_2 & \{1\} & 8km\\
   \sD_{2k+2} & \sD_{2m+2} & D_4 & D_6/{\sim} & 16km\\
   \sD_{lk+2} & \sD_{lm+2} & D_{2l} & \ZZ_l^\times/\left<-1\right> & 8klm,~ (l>2)\\
   \sD_{(2k+1)+2} & \sD_{(2m+1)+2}& \ZZ_4 & \{1\} & 4(2k+1)(2m+1)\\
   \sD_{2(2k+1)+2} & \sD_{2(2m+1)+2}& 2D_4 & D_6/{\sim} & 8(2k+1)(2m+1)\\
   \sD_{l(2k+1)+2} & \sD_{l(2m+1)+2}& 2D_{2l} & \ZZ_{2l}^\times/\left<-1\right> & 4l(2k+1)(2m+1),~(l>2)\\
   \sD_{k+2} & \sD'_{2m+2} & \ZZ_2 & \{1\} & 16km\\
   \sD'_{2k+2} & \sE_7 & \ZZ_2 & \{1\} & 192k\\
   \sD_{k+2} & \sE_7 & \ZZ_2 & \{1\} & 96k\\
   \sD_{3k+2} & \sE_7 & D_6 & \{1\} & 96k\\
   \sE_6 & \sE_6 & \ZZ_3 & \{1\} & 192\\
   \sE_6 & \sE_6 & T & \{1\} & 48\\
   \sE_7 & \sE_7 & \ZZ_2 & \{1\} & 1152\\
   \sE_7 & \sE_7 & D_6 & \{1\} & 384\\
   \sE_7 & \sE_7 & O & \{1\} & 96 \\
   \sE_8 & \sE_8 & I & \ZZ_2 & 240
\end{tabular}
 \label{tab:remaining}
\end{table}

\section{Explicit description of the orbifolds}
\label{sec:expl-desc-orbif}

In an effort to allow comparison with the literature and as a necessary first step in the application of this classification to the problem of identifying the superconformal field theories dual to the corresponding Freund--Rubin backgrounds, we now describe the above orbifolds more explicitly and, whenever possible, in terms of iterated quotients by cyclic groups.

\subsection{Orbifolds as iterated quotients}
\label{sec:orbi-iter-quots}

Let us start by describing iterated quotients.  Let $G$ be a finite group acting effectively and smoothly on a manifold $X$ and let $H<G$ be a subgroup.  Let $Y= X/H$ denote the quotient of $X$ by the action of $H$.  We do not assume that $Y$ is smooth, so it could well be an orbifold.  We can ask whether any elements of $G$ still act on $Y$ via their action on $X$.  Points in $Y$ are equivalence classes $[x]$ of points in $X$, where $x,x' \in X$ are equivalent if $x' = h\cdot x$, for some $h \in H$, hence an element  $g \in G$ still acts on $Y$ provided that the induced action $g \cdot [x] = [g\cdot x]$ is well defined.  This requires that if $[x]=[x']$ then $[g\cdot x] = [g\cdot x']$, which means that if $x' = h\cdot x$ for some $h \in H$, then $g \cdot x' = h' \cdot g \cdot x$ for some $h' \in H$, or equivalently, $(gh) \cdot x = (h'g) \cdot x$ for all $x \in X$.  Since $G$ acts effectively, this implies the identity $gh = h'g$ for all $h$ and some $h'$; that is, $g$ normalises $H$.  In other words, the subgroup of $G$ which acts on $Y$ is the normaliser of $H$ in $G$:
\begin{equation}
  N(H) = \left\{ g \in G \middle | g h g^{-1} \in H~\forall h \in H\right\},
\end{equation}
which is the largest subgroup of $G$ containing $H$ as a normal subgroup.  Of course, since $H$ acts trivially on $Y$, $N(H)$ does not act effectively: it is the factor group $N(H)/H$ which does.  Notice that $N(H)/H$ is \emph{not} a subgroup of $G$.

Now let $K < N(H)/H$ be a subgroup and let us consider the quotient $Y/K$.  We can write this as a quotient of $X$ by a subgroup of $G$: namely the subgroup $L < N(H)$ which maps to $K$ under the natural projection $\pi: N(H) \to N(H)/H$.  In other words, $L$ is an extension of $K$ by $H$:
\begin{equation}
  \label{eq:extensionHLK}
  \begin{CD}
    1 @>>> H @>>> L @>\pi>> K @>>> 1
  \end{CD}
\end{equation}
where the map $L \to K$ is the restriction to $L$ of $\pi$ and given the same name.

Conversely, suppose that $H\lhd L$ is a proper normal subgroup of $L<G$.  Then there is an exact sequence like equation~\eqref{eq:extensionHLK}, where $K = L/H$, and we have an equivalence between the quotient $X/L$ and the iterated quotient $(X/H)/K$.  This can be continued, of course.  Suppose that $H$, which is a subgroup of $L$ and hence of $G$, contains a normal subgroup $N$.  Then we have an exact sequence
\begin{equation}
  \label{eq:extensionNHM}
  \begin{CD}
    1 @>>> N @>>> H @>>> M @>>> 1
  \end{CD}
\end{equation}
with $M \cong H/N$.  Then we have that again
\begin{equation}
  X/H \cong (X/N)/M \implies X/L = (X/N)/(H/N)/(L/H),
\end{equation}
et cetera.

A sequence of subgroups such as $N \lhd H \lhd L$, where $N$ is normal in $H$ (though not necessarily in $L$) and $H$ is normal in $L$, is said to be a \emph{subnormal series} for $L$.  Subnormal series can have any length and more generally, a subnormal series (of length $\ell$) for a group $G$ is a sequence of subgroups
\begin{equation}
  1 = N_0 \lhd N_1 \lhd \cdots \lhd N_\ell = G
\end{equation}
where each $N_i$ is normal in $N_{i+1}$, but not necessarily in $G$.  The quotient $X/G$ is then equivalent to the iterated quotient
\begin{equation}
  (X/N_1)/(N_2/N_1)/(N_3/N_2)/\cdots/(N_\ell/N_{\ell-1})~.
\end{equation}
The groups $N_i/N_{i-1}$ are called the \emph{factors} of the subnormal series.

A group $G$ is \emph{solvable} if it has a subnormal series with abelian factors.  If $G$ is a solvable group, then so is any subgroup and any factor group.  More generally, solvable groups are closed under extension and also under taking products.  All these results are easy to prove and can be found in any book on finite groups, e.g., \cite{MR2014408}.

Now it follows from Table~\ref{tab:normal} that all the ADE subgroups are solvable, with the exception of $\sE_8$.  The relevant subnormal series are easy to determine from that table and are themselves summarised in Table~\ref{tab:subnormal}.  Notice that the subnormal series for $\sE_6$ can be read off from the one for $\sE_7$ since $\sE_6 \lhd \sE_7$.  This table also lists a choice of generators in $\Sp(1)$ for the cyclic factors arising from the subnormal series.  We find it convenient to introduce the following notation for some commonly occurring quaternions:  $\omega_n = e^{2\pi i/n}$, whence $\omega_1=1$, $\omega_2 = -1$, $\omega_4 = i$ and we use $\xi=\omega_8$, and $\zeta = e^{i\pi/4}e^{j\pi/4}$.  In terms of the quaternion generators in Table~\ref{tab:ADE}, $\xi$ corresponds to the element in $\sE_7$ which is not in the normal subgroup $\sE_6$, while $\zeta$ is the element in $\sE_{6,7,8}$ which is not in the normal subgroup $\sD_4$ of $\sE_6$.

\begin{table}[h!]
  \caption{Subnormal series for solvable ADE groups}
  \centering
  \begin{tabular}[t]{>{$}l<{$}|>{$}l<{$}|>{$}l<{$}}
    \multicolumn{1}{c|}{Subnormal series} & \multicolumn{1}{c|}{Factors} & \multicolumn{1}{c}{Generators}\\\hline
    1 \lhd \sA_{n-1} & \ZZ_n & \omega_n\\
   1 \lhd \sA_{2n-1} \lhd \sD_{n+2} & \ZZ_{2n}, \ZZ_2 & \omega_{2n}, j\\
    1 \lhd \sA_3 \lhd \sD_4 \lhd \sE_6 \lhd \sE_7 & \ZZ_4, \ZZ_2, \ZZ_3, \ZZ_2 & i, j, \zeta, \xi
 \end{tabular}
  \label{tab:subnormal}
\end{table}

The cyclic groups $\sA_{n-1}$ are abelian, hence clearly solvable.  The binary dihedral groups $\sD_{n+2}$ have a cyclic normal subgroup $\sA_{2n-1}$ with factor group $\ZZ_2$.  Unlike the dihedral group, $\sD_{n+2}$ (for $n>2$) is not a trivial extension: the only order-2 subgroup of $\sD_{n+2}$ (for $n>2$) is the centre.  The subnormal series $1\lhd \sA_{2n-1} \lhd \sD_{n+2}$ has  cyclic factors $\{\ZZ_{2n},\ZZ_2\}$, showing that $\sD_{n+2}$ is solvable.  The binary tetrahedral group $\sE_6$ has a normal subgroup $\sD_4$ with factor group $\ZZ_3$.   Therefore $1\lhd \sA_3 \lhd \sD_4 \lhd \sE_6$ is a subnormal series with cyclic factors $\{\ZZ_4,\ZZ_2,\ZZ_3\}$.  The binary octahedral group $\sE_7$ has a normal subgroup $\sE_6$ with factor group $\ZZ_2$, whence a subnormal series is obtained from the one of $\sE_6$ by extending it to the right: $1\lhd \sA_3 \lhd \sD_4 \lhd \sE_6 \lhd \sE_7$, with cyclic factors $\{\ZZ_4,\ZZ_2,\ZZ_3,\ZZ_2\}$.  Finally the only proper normal subgroup of the binary icosahedral group is its centre, whence the subnormal series is $1 \lhd \sA_1 \lhd \sE_8$ with factors $\{\ZZ_2,I\}$.  Since $I$ is nonabelian, we conclude that $\sE_8$ is not solvable.

The finite subgroups of $\Spin(4)$ are fibred products of ADE subgroups of $\Sp(1)$, whence they are subgroups of products of ADE subgroups.  Since products and subgroups of solvable groups are solvable, it follows that with the exception of those subgroups involving $\sE_8$, all other subgroups are solvable.  We will see that this means that they have a subnormal series with cyclic factor groups, and hence that the corresponding orbifolds can de described as \emph{iterated cyclic quotients}.  This will simplify their description significantly.  This technique has been used in the context of the $\AdS_5/\text{CFT}_4$ correspondence in \cite{Berenstein:2000mb}.

Recall that the fibred product $A \times_{(F,\tau)} B$ is an extension
\begin{equation}
  \label{eq:extension}
  \begin{CD}
    1 @>>> A_0 \times B_0 @>>> A \times_{(F,\tau)} B @>>> F_\tau @>>> 1,
  \end{CD}
\end{equation}
with $F_\tau = \left\{(x,\tau(x))\middle | x \in F\right\} < F \times F$ the graph of the automorphism $\tau$.  This means that orbifolding by $A \times_{(F,\tau)} B$ can be done in steps: first orbifolding by $A_0 \times B_0$ and then by $F_\tau$.  These quotients can be decomposed further by using subnormal series for $A_0$, $B_0$ and $F_\tau \cong F$.  In the cases of interest, $A_0$ and $B_0$ are ADE subgroups, whose subnormal series have been determined above.  (See Table~\ref{tab:subnormal}.)  It remains to determine subnormal series for the factor groups $F$ in Table~\ref{tab:normal}.

The factor groups $F$ consist of cyclic groups, which need not be decomposed further, dihedral and binary dihedral groups and the tetrahedral, octahedral and icosahedral groups.  This latter group is simple, whence it has no proper normal subgroups.  The dihedral group $D_{2n}$ is a semidirect product $\ZZ_n \rtimes \ZZ_2$, where $\ZZ_n$ is the normal subgroup generated by $y$ in the presentation in equation \eqref{eq:dihedral} and $\ZZ_2$ is the subgroup generated by $x$.  This gives rise to a subnormal series $1 \lhd \ZZ_n \lhd D_{2n}$ with cyclic factors $\{ \ZZ_n, \ZZ_2 \}$.  The binary dihedral group has been discussed previously, being isomorphic to $\sD_{n+2}$.  The subnormal series for the tetrahedral and octahedral groups can be read off from those of their binary cousins, by noticing that $T$ and $O$ are the factor groups obtained by quotienting $\sE_6$ and $\sE_7$, respectively, by their centre, which is contained in any normal subgroup.  Therefore we get an exact sequence $1 \to D_4 \to T \to \ZZ_3 \to 1$ from the similar one involving $\sE_6$, with $D_4 \cong \ZZ_2 \times \ZZ_2$, whereas $1 \to \sE_6 \to \sE_7 \to \ZZ_2 \to 1$ implies $1 \to T \to O \to \ZZ_2 \to 1$.  Hence we get subnormal series $1 \lhd D_4 \lhd T \lhd O$ with abelian factors $\{\ZZ_2 \times \ZZ_2, \ZZ_3,\ZZ_2\}$.  This series can be extended by inserting a $\ZZ_2$ in $1 \lhd D_4$ to obtain $1 \lhd \ZZ_2 \lhd D_4$.  Finally, we truncate at $T$ to get a subnormal series for $T$.  These observations are summarised in Table~\ref{tab:subnormalF}, where we also list a choice of generator for each cyclic factors in $F$, in terms of the generators of $F$ given by the following presentations: $\ZZ_n = \left<z\middle|z^n=1\right>$, $D_{2n} = \left<x,y\middle| x^2=y^n=(x y)^2=1\right>$, $2D_{2n} = \left<s,t\middle|s^2 = t^n = (s t)^2\right>$ and $O=\left<a,b\middle| a^3 = b^4 = (a b)^2 = 1\right>$.

\begin{table}[h!]
  \caption{Subnormal series for factor groups $F$}
  \centering
  \begin{tabular}[t]{>{$}l<{$}|>{$}l<{$}|>{$}l<{$}}
    \multicolumn{1}{c|}{Subnormal series} & \multicolumn{1}{c|}{Factors}& \multicolumn{1}{c}{Generators}\\\hline
    1 \lhd \ZZ_n & \ZZ_n& z\\
    1 \lhd \ZZ_n \lhd D_{2n}  & \ZZ_n, \ZZ_2& y,x\\
    1 \lhd \ZZ_{2n} \lhd 2D_{2n} & \ZZ_{2n},\ZZ_2&t,s\\
    1 \lhd \ZZ_2 \lhd D_4 \lhd T \lhd O & \ZZ_2, \ZZ_2, \ZZ_3, \ZZ_2&b^2, a b^2 a^{-1}, a, b
 \end{tabular}
  \label{tab:subnormalF}
\end{table}

\subsection{Orbifolds involving $\sE_8$}
\label{sec:orbif-involv-e8}

We do not have much to say about these orbifolds at this time.  Here we simply point out that the orbifolds with group $\sE_8 \times_{(I,\tau)} \sE_8$ are in fact $\ZZ_2$-orbifolds of the two smooth quotients given by the graphs of the identity automorphism of $\sE_8$, which has $\eN =5$, and the nontrivial outer automorphism of $\sE_8$, which has $\eN = 4$.  The reason is the following.  As shown in Appendix~\ref{sec:struct-fibr-prod}, the group $\Gamma = \sE_8 \times_{(I,\tau)} \sE_8$ is abstractly isomorphic to $2I \times \ZZ_2$, where the $2I$ subgroup is the graph of any automorphism $\that \in \Aut(2I)$ which induces $\tau \in \Out(I)$. Hence the orbifold $S^7/\Gamma$ is equivalent to the iterated orbifold $S^7/2I/\ZZ_2$, where the first quotient $S^7/2I$ is smooth and where the additional $\ZZ_2$ is generated by $(1,-1) \in \Sp(1) \times \Sp(1)$, say.

\subsection{Solvable orbifolds as iterated cyclic quotients}
\label{sec:solv-orbi-iter}

We will now explicitly discuss the solvable orbifolds in terms of iterated cyclic quotients.  Hence it will be sufficient to indicate the generators of the corresponding cyclic groups.  We recall the notation we use: $\omega_n = e^{2\pi i/n}$, $\xi=\omega_8$ and $\zeta = e^{i\pi/4}e^{j\pi/4}$.

\subsubsection{Orbifolds by product groups}
\label{sec:orbif-prod-groups}

Let us call a group \emph{chiral} if it is contained in one of the two $\Sp(1)$ factors of $\Spin(4)$.  There are four classes of chiral solvable subgroups of $\Sp(1)$: $\sA_{n-1}$, $\sD_{n+2}$, $\sE_6$ and $\sE_7$, and up to conjugation in $\Spin(8)$, they can be taken to belong to the second $\Sp(1)$ factor, say.  The corresponding orbifolds are written as a sequence of iterated cyclic orbifolds as follows:
\begin{itemize}
\item $S^7/\sA_{n-1}$ is the quotient $S^7/\ZZ_n$ by the cyclic group generated by $(1,\omega_n)$.
\item $S^7/\sD_{n+2}$ is the iterated quotient $S^7/\ZZ_{2n}/\ZZ_2$, where the generator of the $\ZZ_{2n}$ action is $(1,\omega_{2n})$ and the generator of the $\ZZ_2$ is $(1,j)$.  Notice that $(1,j)$ has order $4$ in $\Sp(1) \times \Sp(1)$, but it has order $2$ modulo the subgroup generated by $(1,\omega_{2n})$.
\item $S^7/\sE_6$ is equivalent to $S^7/\sD_4/\ZZ_3$, which by the previous case is $S^7/\ZZ_4/\ZZ_2/\ZZ_3$, with generators $(1,i)$, $(1,j)$ and $(1,\zeta)$ in that order.  
\item $S^7/\sE_7$ is equivalent to $S^7/\sE_6/\ZZ_2$, which by the previous case is $S^7/\ZZ_4/\ZZ_2/\ZZ_3/\ZZ_2$, with generators $(1,i)$, $(1,j)$, $(1,\zeta)$ and $(1,\xi)$.
\end{itemize}

Chiral groups are special cases of product groups, those which are the product of two chiral groups (of opposite chirality).  The product groups are listed in Table~\ref{tab:products}.  The results from the chiral groups can be used to rewrite orbifolds by product groups as iterated quotients.  As an illustration, let us consider the orbifold $S^7/(\sD_{n+2}\times \sE_6)$.  This is the iterated quotient $S^7/\ZZ_{2n}/\ZZ_2/\ZZ_4/\ZZ_2/\ZZ_3$, where the cyclic generators are, in the order written, given by $(\omega_{2n},1)$, $(j,1)$, $(1,i)$, $(1,j)$ and $(1,\zeta)$.  Of course, there is some freedom in the order in which we have written the generators.  We would obtain an equivalent orbifold if were to shuffle the generators in such a way that generators belonging to the same chiral $\Sp(1)$ are kept in the same order.  Table~\ref{tab:productsiterated} lists the orbifolds associated to product solvable subgroups of $\Spin(4)$ and expresses them as iterated cyclic orbifolds.  They all have $\eN=4$ supersymmetry.

\begin{table}[h!]
  \caption{Orbifolds by solvable product groups as iterated cyclic quotients}
  \centering
  \begin{tabular}[t]{>{$}l<{$}|>{$}l<{$}|>{$}l<{$}}
    \multicolumn{1}{c|}{$\Gamma$} & \multicolumn{1}{c|}{$S^7/\Gamma$} & \multicolumn{1}{c}{Cyclic generators}\\\hline
    \sA_{n-1} \times \sA_{m-1} & S^7/\ZZ_n/\ZZ_m & (\omega_n,1),(1,\omega_m) \\
    \sA_{n-1} \times \sD_{m+2} & S^7/(\sA_{n-1}\times\sA_{2m-1})/\ZZ_2 & (1,j)\\
    \sA_{n-1} \times \sE_6 & S^7/(\sA_{n-1}\times\sD_4)/\ZZ_3 & (1,\zeta)\\
    \sA_{n-1} \times \sE_7 & S^7/(\sA_{n-1}\times \sE_6)/\ZZ_2 & (1,\xi)\\
    \sD_{n+2} \times \sD_{m+2} & S^7/(\sA_{2n-1} \times \sD_{m+2})/\ZZ_2 & (j,1)\\
    \sD_{n+2} \times \sE_6 & S^7/(\sD_{n+2} \times \sD_4)/\ZZ_3 & (1,\zeta)\\
    \sD_{n+2} \times \sE_7 & S^7/(\sD_{n+2} \times \sE_6)/\ZZ_2 & (1,\xi)\\
    \sE_6 \times \sE_6 & S^7/(\sD_4 \times \sE_6)/\ZZ_3 & (1,\zeta)\\
    \sE_6 \times \sE_7 & S^7/(\sE_6 \times \sE_6)/\ZZ_2 & (1,\xi)\\
    \sE_7 \times \sE_7 & S^7/(\sE_6 \times \sE_7)/\ZZ_2 & (\xi,1)
  \end{tabular}
  \label{tab:productsiterated}
\end{table}

\subsubsection{Smooth quotients}
\label{sec:smooth-quotients}

We now turn our attention to the smooth quotients by solvable groups in Table~\ref{tab:smooth}.  These include all but the two quotients associated to $\sE_8$ and mentioned already in Section~\ref{sec:orbif-involv-e8}.

The smooth quotients $S^7/\sA_{n-1}\times_{(\ZZ_n,\tau)}\sA_{n-1} $ by cyclic groups are generated by the element $(\omega_n, \omega_n^q)$ in $\Sp(1) \times \Sp(1)$, where $q$ is coprime to $n$.  Because the subgroups corresponding to $q$ and $-q$ are conjugate in $\Sp(1)\times\Sp(1)$ --- they are related by conjugation by $(1,j)$ --- the corresponding orbifolds are equivalent.  For $q = \pm 1$, the quotient has $\eN=8$ for $n=2$, $\eN = 6$ for $n>2$, whereas for any other value of $q$, it has $\eN=4$.

The smooth quotients $S^7/\sD_{n+2} \times_{(2D_{2n},\tau)} \sD_{n+2}$ by binary dihedral groups can be described as an iterated quotient by cyclic groups: $S^7/\ZZ_{2n}/\ZZ_2$, where the action of $\ZZ_{2n}$ is generated by the element $(\omega_{2n},\omega_{2n}^r)$, where $r$ is coprime to $2n$, and the action of $\ZZ_2$ is generated by $(j,j)$.  Again we identify $r$ with $-r$, since they are conjugate by $(1,j)$ in $\Sp(1)\times \Sp(1)$.  If $r=\pm 1$ the quotient has $\eN=5$ and it is a further $\ZZ_2$ quotient of the smooth $\eN=6$ lens space $S^7/\ZZ_{2n}$, otherwise the quotient has $\eN=4$ and it is a further $\ZZ_2$ quotient of the smooth $\eN=4$ quotient by $\ZZ_{2n}$.  When $n=2$, the outer automorphism group of $2D_4$ is $D_6$, but as discussed in Section~\ref{sec:outer-autom-2d_2n}, all automorphisms are induced by conjugation inside an $\sE_7$ subgroup of $\Sp(1)$ and hence the associated quotients are equivalent to the $\eN=5$ quotient corresponding to the identity automorphism.  In other words, $\Twist(2D_4) \cong \{1\}$ in this case.  We remind the reader that the notation $\Twist(F)$ hides the fact that it is not just a function of the abstract group $F$, but depends on the way that $F$ can be obtained as a factor group of $A$ and $B$.

The smooth $\eN=5$ quotient $S^7/\sE_6 \times_{2T} \sE_6$ is equivalent to $S^7/\ZZ_4/\ZZ_2/\ZZ_3$, where the action of $\ZZ_4$ is generated by the element $(i,i)$, that of $\ZZ_2$ by $(j,j)$, and the one of $\ZZ_3$ by $(\zeta,\zeta)$.  Hence $S^7/2T$ can be understood as the $\ZZ_3$ quotient  of the $\eN=5$ quotient $S^7/\sD_4$ or as a $D_6$ quotient of the $\eN=6$ lens space $S^7/\ZZ_4$.  This latter quotient can be done in two steps: the first a quotient by $\ZZ_2$ and the second by $\ZZ_3$, as indicated.

Finally, the quotient $S^7/\sE_7 \times_{(2O,\tau)}\sE_7$ is equivalent to $S^7/\ZZ_4/\ZZ_2/\ZZ_3/\ZZ_2$, where the action of $\ZZ_4$ is generated by the element $(i,i)$, the first $\ZZ_2$ action by $(j,j)$, the $\ZZ_3$ action by $(\zeta,\zeta)$ and the final $\ZZ_2$ action by $(\xi,\pm\xi)$, where the signs $\pm$ correspond to the two quotients: the $+$ sign for the diagonal $\eN=5$ quotient and the $-$ sign for the twisted $\eN=4$ quotient.

Table~\ref{tab:smoothiterated} summarises the smooth quotients by solvable groups, expressed as iterated cyclic quotients. In that table, $q \in \ZZ_n^\times/\left<-1\right>$ and $r \in \ZZ_{2n}^\times/\left<-1\right>$.

\begin{table}[h!]
  \caption{Smooth quotients by solvable groups as iterated cyclic quotients}
  \centering
  \begin{tabular}[t]{>{$}l<{$}|>{$}l<{$}|>{$}l<{$}}
    \multicolumn{1}{c|}{$\Gamma$} & \multicolumn{1}{c|}{$S^7/\Gamma$} & \multicolumn{1}{c}{Cyclic generators}\\\hline
    \sA_{n-1}\times_{(\ZZ_n,\tau)}\sA_{n-1} & S^7/\ZZ_n & (\omega_n,\omega_n^q)\\
    \sD_{n+2} \times_{(2D_{2n},\tau)} \sD_{n+2}& S^7/\ZZ_{2n}/\ZZ_2 & (\omega_{2n},\omega_{2n}^r), (j,j)\\
    \sE_6 \times_{2T} \sE_6 & S^7/\ZZ_4/\ZZ_2/\ZZ_3 & (i,i),(j,j),(\zeta,\zeta) \\
    \sE_7 \times_{(2O,\tau)}\sE_7 & S^7/\ZZ_4/\ZZ_2/\ZZ_3/\ZZ_2 & (\zeta\xi,\zeta\xi),(i,i),(\zeta,\zeta),(\xi,\pm\xi)
  \end{tabular}
  \label{tab:smoothiterated}
\end{table}

\subsubsection{Remaining solvable orbifolds}
\label{sec:rema-solv-orbif}

It remains to discuss those orbifolds in Table~\ref{tab:remaining} where the group $A\times_{(F,\tau)}B$ is solvable; that is, all cases but the order-240 groups $\sE_8 \times_{(I,\tau)} \sE_8$.  As discussed above, the group $A\times_{(F,\tau)}B$ can be written as an extension \eqref{eq:extension}, in which case the corresponding orbifold is $S^7/(A_0 \times B_0)/F_\tau$, where $F_\tau$ may decompose further into cyclic groups as summarised in Table \ref{tab:subnormalF}.

Let us first treat those cases where $F_\tau$ is already a cyclic group.  Then a possible generator for the action of $F_\tau$ is $(x,\tau(x)) \in \Gamma$, where $x$ is read off from the cyclic generator of $F$ in Table~\ref{tab:subnormal}.

\begin{itemize}
\item $\Gamma = \sA_{kl-1}\times_{(\ZZ_l,\tau)} \sA_{ml-1}$\\
Here $\tau$ is the automorphism corresponding to multiplication by $r \in \ZZ_l^\times$.  The extension takes the form
\begin{equation}
  \begin{CD}
    1 @>>> \sA_{k-1} \times \sA_{m-1} @>>>\Gamma @>>> (\ZZ_l)_\tau @>>> 1,
  \end{CD}
\end{equation}
where $(\ZZ_l)_\tau$ is the subgroup of $\ZZ_l \times \ZZ_l$ consisting of pairs $(a,b)$ with $b = ra \pmod l$ and where the maps $\sA_{k-1} \to \sA_{kl-1}$ and  $\sA_{m-1} \to \sA_{ml-1}$ consist of taking the $l$th power of the generator.  We may describe this orbifold as an iterated quotient in the following way: $S^7/\ZZ_k/\ZZ_m/\ZZ_l$, with cyclic generators $(\omega_k,1)$, $(1,\omega_m)$ and $(\omega_{kl},\omega_{ml}^r)$.  The special case of $r=1$ corresponds to the orbifolds discussed in \cite[§3.4]{Imamura:2008nn} and \cite[§2]{Imamura:2008ji} and to be discussed in Section~\ref{sec:conclusion}.   Notice that the automorphisms labelled by $r$ and $-r$ are actually related by conjugation by the element $(1,j)$ in $\Sp(1)\times \Sp(1)$, whence the two orbifolds are actually equivalent.

\item $\Gamma = \sA_{2k-1} \times_{\ZZ_2} \sD'_{2m+2}$\\
This group is the extension
\begin{equation}
  \begin{CD}
    1 @>>> \sA_{k-1} \times \sD_{m+2} @>>>\Gamma @>>> \ZZ_2 @>>> 1.
  \end{CD}
\end{equation}
Therefore, the orbifold $S^7/\Gamma$ can be described by $S^7/(\sA_{k-1} \times \sD_{m+2})/\ZZ_2$, where the first quotient is by a product group and hence already discussed in Section \ref{sec:orbif-prod-groups}, and the $\ZZ_2$ action is generated by the element $(\omega_{2k},\omega_{4m})$ which has order 2 mod $\sA_{k-1} \times \sD_{m+2}$.

\item $\Gamma = \sA_{2k-1} \times_{\ZZ_2} \sD_{m+2}$\\
This group is the extension
\begin{equation}
  \begin{CD}
    1 @>>> \sA_{k-1} \times \sA_{2m-1} @>>>\Gamma @>>> \ZZ_2 @>>> 1.
  \end{CD}
\end{equation}
Therefore, the orbifold $S^7/\Gamma$ can be described by $S^7/(\sA_{k-1} \times \sA_{2m-1})/\ZZ_2$, where the first quotient is by a product group and hence already discussed in Section \ref{sec:orbif-prod-groups}, and the $\ZZ_2$ action is generated by the element $(\omega_{2k},j)$ which has order 2 mod $\sA_{k-1} \times \sA_{2m-1}$.

\item $\Gamma = \sA_{4k-1} \times_{(\ZZ_4,\tau)} \sD_{2m+3}$\\
This group is the extension
\begin{equation}
  \begin{CD}
    1 @>>> \sA_{k-1} \times \sA_{2m} @>>>\Gamma @>>> (\ZZ_4)_\tau @>>> 1,
  \end{CD}
\end{equation}
where $\tau \in \ZZ_4^\times$.  Therefore, the orbifold $S^7/\Gamma$ can be described by $S^7/(\sA_{k-1} \times \sA_{2m})/\ZZ_4$, where the first quotient is by a product group, and the $\ZZ_4$ action is generated by the element $(\omega_{4k},\pm j)$, where the choice of sign is the choice of $\tau \in \ZZ_4^\times$.  Both choices are related by conjugation by $(1,i)$, whence they define equivalent twists.

\item $\Gamma = \sA_{3k-1} \times_{(\ZZ_3,\tau)} \sE_6$\\
Here $\tau \in \ZZ_3^\times$ is an automorphism which will manifest itself in a choice of sign.   This group is an extension
\begin{equation}
  \begin{CD}
    1 @>>> \sA_{k-1} \times \sD_4 @>>>\Gamma @>>> \ZZ_3 @>>> 1,
  \end{CD}
\end{equation}
whence the orbifold $S^7/\Gamma$ is equivalent to $S^7/(\sA_{k-1} \times \sD_4)/\ZZ_3$, where the first orbifold is by a product group, hence already discussed, and the action of the $\ZZ_3$ is generated by the element $(\omega_{3k}^{\pm 1}, \zeta)$, where the $+$ sign corresponds to the trivial automorphism and the $-$ sign to the nontrivial automorphism of $\ZZ_3$.  The two signs are related by conjugation by the element $(j,1)\in \Sp(1)\times\Sp(1)$, which normalises the subgroup $\sA_{k-1} \times \sD_4$, whence both orbifolds are equivalent.

\item $\Gamma = \sA_{2k-1} \times_{\ZZ_2} \sE_7$\\
This group is the extension
\begin{equation}
  \begin{CD}
    1 @>>> \sA_{k-1} \times \sE_6 @>>>\Gamma @>>> \ZZ_2 @>>> 1,
  \end{CD}
\end{equation}
whence the orbifold $S^7/\Gamma$ can be described as $S^7/(\sA_{k-1} \times \sE_6)/\ZZ_2$, where the first orbifold is by a product group and the action of the $\ZZ_2$ is generated by the element $(\omega_{2k},\xi)$.

\item $\Gamma = \sD'_{2k+2} \times_{\ZZ_2} \sD'_{2m+2}$\\
This group is the extension
\begin{equation}
  \begin{CD}
    1 @>>> \sD_{k+2} \times \sD_{m+2} @>>>\Gamma @>>> \ZZ_2 @>>> 1,
  \end{CD}
\end{equation}
whence the orbifold $S^7/\Gamma$ can be described as $S^7/(\sD_{k+2} \times \sD_{m+2})/\ZZ_2$, where the first orbifold is by a product group and the action of $\ZZ_2$ is generated by the element $(\omega_{4k},\omega_{4m})$.

\item $\Gamma = \sD_{k+2} \times_{\ZZ_2} \sD_{m+2}$\\
This group is the extension
\begin{equation}
  \begin{CD}
    1 @>>> \sA_{2k-1} \times \sA_{2m-1} @>>>\Gamma @>>> \ZZ_2 @>>> 1,
  \end{CD}
\end{equation}
whence the orbifold $S^7/\Gamma$ can be described as $S^7/(\sA_{2k-1} \times \sA_{2m-1})/\ZZ_2$, where the action of $\ZZ_2$ is generated by the element $(j,j)$.

\item $\Gamma = \sD_{(2k+1)+2} \times_{(\ZZ_4,\tau)} \sD_{(2m+1)+2}$\\
This group is the extension
\begin{equation}
  \begin{CD}
    1 @>>> \sA_{2k} \times \sA_{2m} @>>>\Gamma @>>> (\ZZ_4)_\tau @>>> 1,
  \end{CD}
\end{equation}
whence the orbifold $S^7/\Gamma$ can be described as $S^7/(\sA_{2k} \times \sA_{2m})/\ZZ_4$, where the action of $\ZZ_4$ is generated by the element $(j,\pm j)$, where the choice of sign corresponds to the choice of $\tau \in \ZZ_4^\times$.  Clearly both choices are conjugate via $(1,i)$ and hence give rise to equivalent twists.

\item $\Gamma = \sD_{k+2} \times_{\ZZ_2} \sD'_{2m+2}$\\
This group is the extension
\begin{equation}
  \begin{CD}
    1 @>>> \sA_{2k-1} \times \sD_{m+2} @>>>\Gamma @>>> \ZZ_2 @>>> 1,
  \end{CD}
\end{equation}
whence the orbifold $S^7/\Gamma$ can be described as $S^7/(\sA_{2k} \times \sD_{m+2})/\ZZ_2$, where the action of $\ZZ_2$ is generated by the element $(j,\omega_{4m})$.

\item $\Gamma = \sD_{k+2} \times_{\ZZ_2} \sE_7$\\
This group is the extension
\begin{equation}
  \begin{CD}
    1 @>>> \sA_{2k-1} \times \sE_6 @>>>\Gamma @>>> \ZZ_2 @>>> 1,
  \end{CD}
\end{equation}
whence the orbifold $S^7/\Gamma$ can be described as $S^7/(\sA_{2k-1} \times \sE_6)/\ZZ_2$, where the first orbifold is by a product group and the action of the $\ZZ_2$ is generated by the element $(j,\xi)$.

\item $\Gamma = \sD'_{2k+2} \times_{\ZZ_2} \sE_7$\\
This group is the extension
\begin{equation}
  \begin{CD}
    1 @>>> \sD_{k+2} \times \sE_6 @>>>\Gamma @>>> \ZZ_2 @>>> 1,
  \end{CD}
\end{equation}
whence the orbifold $S^7/\Gamma$ can be described as $S^7/(\sD_{k+2} \times \sE_6)/\ZZ_2$, where the first orbifold is by a product group and the action of the $\ZZ_2$ is generated by the element $(\omega_{4k},\xi)$.

\item $\Gamma = \sE_6 \times_{(\ZZ_3,\tau)} \sE_6$\\
This group is the extension
\begin{equation}
  \begin{CD}
    1 @>>> \sD_4 \times \sD_4 @>>>\Gamma @>>> \ZZ_3 @>>> 1,
  \end{CD}
\end{equation}
whence $S^7/\Gamma$ is given by the iterated orbifold $S^7/(\sD_4 \times \sD_4)/\ZZ_3$, where the first quotient is by a product group, and the action of the $\ZZ_3$ is given by $(\zeta,\zeta)$ or, if $\tau$ is nontrivial, $(\zeta,\tau(\zeta))$, where $\tau(\zeta) = e^{j\pi/4}e^{i\pi/4}$.   Both choices are conjugate in $\Sp(1) \times \Sp(1)$ by the element $(1,\zeta\xi)$, whence we are free to take $\tau=1$.

\item $\Gamma = \sE_7 \times_{\ZZ_2} \sE_7$\\
This group is the extension
\begin{equation}
  \begin{CD}
    1 @>>> \sE_6 \times \sE_6 @>>>\Gamma @>>> \ZZ_2 @>>> 1,
  \end{CD}
\end{equation}
whence $S^7/\Gamma$ is a $\ZZ_2$ quotient of $S^7/\sE_6 \times \sE_6$, where the $\ZZ_2$-action is generated by $(\xi,\xi)$.
\end{itemize}

Finally we consider the more complicated cases where $F_\tau$ is not cyclic and must be decomposed further.  Here we need to find elements in $\Gamma$ which project down to the generators of the cyclic groups into which we decompose $F_\tau$.  In all cases, this can be done using the information in Table~\ref{tab:subnormalF}.

\begin{itemize}
\item $\Gamma =\sD_{2kl+2} \times_{(D_{4l},\tau)} \sD_{2ml+2}$ ($l>1$) and $\Gamma = \sD_{k(2l+1)+2} \times_{(D_{4l+2},\tau)} \sD_{m(2l+1)+2}$ ($l>0$)\\
These two groups are special cases of $\Gamma = \sD_{kl+2} \times_{(D_{2l},\tau)} \sD_{ml+2}$, which is the extension
\begin{equation}
  \begin{CD}
    1 @>>> \sA_{2k-1} \times \sA_{2m-1} @>>>\Gamma @>>> D_{2l} @>>> 1,
  \end{CD}
\end{equation}
whence the orbifold $S^7/\Gamma$ can be described as $S^7/(\sA_{2k-1} \times \sA_{2m-1})/\ZZ_l/\ZZ_2$, where the first orbifold is by a product group, the action of $\ZZ_l$ is generated by the element $(\omega_{2kl},\omega_{2ml}^r)$, for $r \in \ZZ_l^\times/\left<-1\right>$, and the action of $\ZZ_2$ is generated by $(j,j)$.

\item $\Gamma =\sD_{2k+2} \times_{(D_4,\tau)} \sD_{2m+2}$\\
This group is the extension
\begin{equation}
  \begin{CD}
    1 @>>> \sA_{2k-1} \times \sA_{2m-1} @>>>\Gamma @>>> (D_4)_\tau @>>> 1,
  \end{CD}
\end{equation}
where $\tau \in \Out(D_4) \cong D_6$.  One way to think about $\tau$ is an identification (consistent with the group multiplication) between the cosets of $\sA_{2k-1}$ in $\sD_{2k+2}$ and those of $\sA_{2m-1}$ in $\sD_{2m+2}$.  Since $\sD_{2k+2}= \left<j,\omega_{4k}\right>$ and $\sA_{2k-1} = \left<\omega_{2k}\right>$, the elements of $\sD_{2k+2}/\sA_{2k-1}$ are given by the cosets $[1]$, $[\omega_{4k}]$, $[j]$ and $[j\omega_{4k}]$.  Since $-1\in [1]$, all non-identity elements of $\sD_{2k+2}/\sA_{2k-1}$ have order 2.  Therefore any bijection between the two sets $\{[\omega_{4k}],[j],[j\omega_{4k}] \}$ and $\{[\omega_{4m}],[j],[j\omega_{4m}] \}$ gives a possible $\tau$.  Conjugation by $\omega_{8k}$ in $\Sp(1)$ defines an outer automorphism of $\sD_{2k+2}$ of order 2 which exchanges $[j] \leftrightarrow [j\omega_{4k}]$ and similarly for $k$ replaced by $m$.  The set of inequivalent twists is the set of orbits of $\Out(F) \cong D_6$ under this action of $\ZZ_2 \times \ZZ_2$, and it is not hard to see that this set consists of two elements: the identity and the automorphism which exchanges $[j]$ and $[\omega_{4k}]$ in $\sD_{2k+2}/\sA_{2k-1}$.  In summary, there are two in equivalent orbifolds $S^7/\Gamma$: one equivalent to $S^7/(\sA_{2k-1} \times \sA_{2m-1})/\ZZ_2/\ZZ_2$, where the generators of the $\ZZ_2$-actions are $(j,j)$ and $(\omega_{4k},\omega_{4m})$, respectively; and the other also equivalent to $S^7/(\sA_{2k-1} \times \sA_{2m-1})/\ZZ_2/\ZZ_2$, but where the generators of the $\ZZ_2$-actions are now $(j,\omega_{4m})$ and $(\omega_{4k},j)$.

\item $\Gamma = \sD_{l(2k+1)+2} \times_{(2D_{2l},\tau)} \sD_{l(2m+1)+2}$ ($l>2$)\\
This group is the extension
\begin{equation}
  \begin{CD}
    1 @>>> \sA_{2k} \times \sA_{2m} @>>>\Gamma @>>> 2D_{2l} @>>> 1,
  \end{CD}
\end{equation}
whence the orbifold $S^7/\Gamma$ is equivalent to $S^7/(\sA_{2k} \times \sA_{2m})/\ZZ_{2l}/\ZZ_2$, where the $\ZZ_{2l}$-action is generated by $(\omega_{2l(2k+1)},\omega_{2l(2m+1)}^r)$, for $r \in \ZZ_{2l}^\times/\left<-1\right>$, and the action of $\ZZ_2$ is generated by $(j,j)$.

\item $\Gamma = \sD_{2(2k+1)+2} \times_{(2D_4,\tau)} \sD_{2(2m+1)+2}$\\
This case is very similar to that of $\sD_{2k+2} \times_{(D_4,\tau)} \sD_{2m+2}$.  The extension now is
\begin{equation}
  \begin{CD}
    1 @>>> \sA_{2k} \times \sA_{2m} @>>>\Gamma @>>> (2D_4)_\tau @>>> 1,
  \end{CD}
\end{equation}
where $\tau \in \Out(2D_4)\cong D_6$ again.  The group $2D_4$ is the quaternion group, so abstractly isomorphic to $\sD_4$.  The details of the possible twists are \emph{mutatis mutandis} like in the case of  $\sD_{2k+2} \times_{(D_4,\tau)} \sD_{2m+2}$, except that since $-1\not [1]$ we need to take signs into account.  The signs take care of themselves, however, and the upshot is that there are again two inequivalent twists and hence two inequivalent quotients $S^7/\Gamma$, equivalent to $S^7/ (\sA_{2k} \times \sA_{2m})/\ZZ_4/\ZZ_2$, where in one case the generators of $\ZZ_4$ and $\ZZ_2$ are given, respectively, by $(j,j)$ and $(\omega_{4(2k+1)},\omega_{4(2m+1)})$, and in the other case by $(j,\omega_{4(2m+1)})$ and $(\omega_{4(2m+1)}, j)$.

\item $\Gamma = \sD_{3k+2} \times_{(D_6,\tau)} \sE_7$\\
This group is the extension
\begin{equation}
  \begin{CD}
    1 @>>> \sA_{2k-1} \times \sD_4 @>>>\Gamma @>>> D_6 @>>> 1,
  \end{CD}
\end{equation}
whence the orbifold $S^7/\Gamma$ can be described as $S^7/(\sA_{2k-1} \times \sD_4)/\ZZ_3/\ZZ_2$, where the first orbifold is by a product group and the actions of $\ZZ_3$ and $\ZZ_2$ are generated, respectively, by the elements $(\omega_{6k},\zeta)$ and $(j,\xi)$.

\item $\Gamma = \sE_6 \times_{(T,\tau)} \sE_6$\\
The automorphism $\tau$ can be taken to be the identity without loss of generality, since the nontrivial outer automorphism of $T$ is induced by conjugation in $\Sp(1)$ (see, e.g., \cite{deMedeiros:2009pp}).  As discussed in Appendix~\ref{sec:struct-fibr-prod}, this is a $\ZZ_2$ orbifold of the $\eN=5$ smooth quotient by $\sE_6$, described in \ref{sec:smooth-quotients}, where the action of $\ZZ_2$ group is generated by the element $(1,-1)$.  Alternatively, $\Gamma$ is the extension
\begin{equation}
  \begin{CD}
    1 @>>> \sA_1 \times \sA_1 @>>>\Gamma @>>> T @>>> 1,
  \end{CD}
\end{equation}
whence $S^7/\Gamma$ is equivalent to $S^7/(\sA_1 \times \sA_1)/\ZZ_2/\ZZ_2/\ZZ_3$, where the actions of the last three cyclic groups are generated by the elements $(i,i)$, $(j,j)$ and $(\zeta,\zeta)$.

\item $\Gamma = \sE_7 \times_{O} \sE_7$\\
As discussed in Appendix~\ref{sec:struct-fibr-prod}, this is again a $\ZZ_2$ orbifold of the $\eN=5$ smooth quotient by $\sE_7$, described in \ref{sec:smooth-quotients}, where the action of $\ZZ_2$ group is generated by the element $(1,-1)$.  Alternatively, it is a further $\ZZ_2$ quotient of the previous case.  Indeed, $\Gamma$ is the extension
\begin{equation}
  \begin{CD}
    1 @>>> \sA_1 \times \sA_1 @>>>\Gamma @>>> O @>>> 1,
  \end{CD}
\end{equation}
whence we could as well describe the orbifold as $S^7/(\sA_1 \times \sA_1)/\ZZ_2/\ZZ_2/\ZZ_3/\ZZ_2$, where the actions of the last four cyclic groups are generated by the elements $(i,i)$, $(j,j)$, $(\zeta,\zeta)$ and $(\xi,\xi)$.

\item $\Gamma = \sE_7 \times_{D_6} \sE_7$\\
This is simply a further $\ZZ_2$ orbifold of the previous case.  Indeed, the group is an extension
\begin{equation}
  \begin{CD}
    1 @>>> \sD_4 \times \sD_4 @>>>\Gamma @>>> D_6 @>>> 1,
  \end{CD}
\end{equation}
whence $S^7/\Gamma$ is a $\ZZ_2$ quotient of $S^7/(\sE_6 \times_{\ZZ_3} \sE_6)$, where the $\ZZ_2$ action is generated by $(\xi, \xi)$.

\end{itemize}

Table~\ref{tab:remainingiterated} summarises these results.  For lack of space, we only list cyclic generators corresponding to the quotient by $F$, since the ones corresponding to the quotient by the product group can be read off from Table~\ref{tab:productsiterated}.  We recall the definitions $\omega_n = e^{2\pi i/n}$, $\xi = \omega_8$ and $\zeta = e^{i\pi/4}e^{j\pi/4}$.  In the first and tenth lines, $q \in \ZZ_l^\times/\left<-1\right>$, i.e., an integer modulo $l$ and coprime to $l$ and where $q$ and $-q$ are identified, while in the thirteenth line $r \in \ZZ_{2l}^\times/\left<-1\right>$.

\begin{table}[h!]
  \caption{Remaining orbifolds by solvable groups as iterated cyclic quotients}
  \centering
  \begin{tabular}[t]{>{$}l<{$}|>{$}l<{$}|>{$}l<{$}}
    \multicolumn{1}{c|}{$\Gamma = A \times_{(F,\tau)} B$} & \multicolumn{1}{c|}{$S^7/\Gamma = S^7/(A_0 \times B_0)/F$} & \multicolumn{1}{c}{Cyclic generators for $F$}\\\hline
    \sA_{kl-1} \times_{(\ZZ_l,\tau)} \sA_{ml-1} & S^7/(\sA_{k-1}\times \sA_{m-1})/\ZZ_l & (\omega_{kl},\omega_{ml}^q)\\
    \sA_{2k-1} \times_{\ZZ_2} \sD'_{2l+2} & S^7/(\sA_{k-1} \times \sD_{l+2})/\ZZ_2 & (\omega_{2k},\omega_{4l})\\
    \sA_{2k-1} \times_{\ZZ_2} \sD_{m+2} & S^7/(\sA_{k-1} \times \sA_{2m-1})/\ZZ_2 & (\omega_{2k},j)\\
    \sA_{4k-1} \times_{\ZZ_4} \sD_{2m+3} & S^7/(\sA_{k-1} \times \sA_{2m})/\ZZ_4 & (\omega_{4k},j)\\
    \sA_{3k-1} \times_{\ZZ_3} \sE_6 & S^7/(\sA_{k-1} \times \sD_4)/\ZZ_3 & (\omega_{3k},\zeta)\\
    \sA_{2k-1} \times_{\ZZ_2} \sE_7 & S^7/(\sA_{k-1}\times \sE_6)/\ZZ_2 & (\omega_{2k},\xi)\\
    \sD'_{2k+2} \times_{\ZZ_2} \sD'_{2l+2} & S^7/(\sD_{k+2}\times\sD_{l+2})/\ZZ_2 & (\omega_{4k},\omega_{4l})\\
    \sD_{k+2} \times_{\ZZ_2} \sD_{m+2} & S^7/(\sA_{2k-1} \times \sA_{2m-1})/\ZZ_2 & (j,j)\\
    \sD_{2k+2} \times_{(D_4,\tau)} \sD_{2m+2} & S^7/(\sA_{2k-1}\times\sA_{2m-1})/\ZZ_2/\ZZ_2& \begin{cases} (j,j),(\omega_{4k},\omega_{4m})\\ (j,\omega_{4m}),(\omega_{4k},j)\end{cases}\\
   \sD_{kl+2} \times_{(D_{2l},\tau)} \sD_{ml+2} & S^7/(\sD_{k+2} \times \sD_{m+2}) /\ZZ_l/\ZZ_2 & (\omega_{2kl},\omega_{2ml}^q),(j,j)\\
    \sD_{(2k+1)+2} \times_{\ZZ_4} \sD_{(2m+1)+2} & S^7/(\sA_{2k} \times \sA_{2m})/\ZZ_4 & (j,j) \\
   \sD_{2(2k+1)+2} \times_{(2D_4,\tau)} \sD_{2(2m+1)+2} & S^7/(\sA_{2k}\times\sA_{2m})/\ZZ_4/\ZZ_2& \begin{cases} (j,j), (\omega_{4(2k+1)},\omega_{4(2m+1)})\\ (j,\omega_{4(2m+1)}), (\omega_{4(2m+1)}, j)\end{cases} \\
   \sD_{l(2k+1)+2} \times_{(2D_{2l},\tau)} \sD_{l(2m+1)+2} & S^7/(\sA_{2k}\times\sA_{2m})/\ZZ_{2l}/\ZZ_2 &  (\omega_{2l(2k+1)},\omega_{2l(2m+1)}^r),(j,j)\\
    \sD_{k+2} \times_{\ZZ_2} \sD'_{2m+2} & S^7/(\sA_{2k-1} \times \sD_{m+2})/\ZZ_2 & (\omega_{2k},\omega_{4m})\\
   \sD'_{2k+2} \times_{\ZZ_2} \sE_7 &  S^7/(\sD_{k+2}\times \sE_6)/\ZZ_2 & (\omega_{4k},\xi)\\
    \sD_{k+2} \times_{\ZZ_2} \sE_7 & S^7/(\sA_{2k-1} \times \sE_6)/\ZZ_2 & (j,\xi)\\
    \sD_{3k+2} \times_{D_6} \sE_7 & S^7/(\sA_{2k-1}\times \sD_4)/\ZZ_3/\ZZ_2 & (\omega_{6k},\zeta), (j,\xi)\\
    \sE_6 \times_{\ZZ_3} \sE_6 & S^7/(\sD_4 \times \sD_4)/\ZZ_3 & (\zeta,\zeta)\\
    \sE_6 \times_{T} \sE_6 & S^7/(\sA_1 \times \sA_1)/\ZZ_2/\ZZ_2/\ZZ_3 & (i,i),(j,j),(\zeta,\zeta)\\
    \sE_7 \times_{\ZZ_2} \sE_7 & S^7/(\sE_6 \times \sE_6)/\ZZ_2 & (\xi,\xi)\\
    \sE_7 \times_{D_6} \sE_7 & S^7/(\sD_4 \times \sD_4)/\ZZ_3/\ZZ_2 & (\zeta,\zeta),(\xi,\xi)\\
    \sE_7 \times_{O} \sE_7 & S^7/(\sA_1 \times \sA_1)/\ZZ_2/\ZZ_2/\ZZ_3/\ZZ_2 & (i,i),(j,j),(\zeta,\zeta),(\xi,\xi)
 \end{tabular}
  \label{tab:remainingiterated}
\end{table}

\section{Conclusion and outlook}
\label{sec:conclusion}

By way of conclusion, we shall use the results from the previous section to identify a few of the $\eN \geq 4$ orbifolds classified in this paper which have been encountered already in the recent M2-brane literature --- the intention being to thereby motivate the investigation of M2-brane interpretations for the many new $\eN = 4$ orbifolds we have found. All the known examples in this context have been obtained from moduli spaces for certain $\eN \geq 4$ superconformal field theories in three dimensions. As established in \cite{SCCS3Algs}, all such theories necessarily involve a non-dynamical gauge field described by a Chern--Simons term coupled to matter fields which fall into hypermultiplet representations of the $\eN = 4$ conformal superalgebra in three dimensions. Typically the superconformal moduli spaces of gauge-inequivalent vacua in theories of this type have a rather complicated structure though they invariably contain a particular branch involving constant matter fields which parametrise the transverse space to the M2-branes, on which the superconformal field theories furnish a low-energy effective description. Following the standard recipe for holographic duality in the $\AdS_4/\text{CFT}_3$ context, in the strong coupling limit, one identifies this branch in the $\text{CFT}_3$ moduli space with the metric cone $\RR^8/\Gamma$ over the quotient $S^7/\Gamma$ appearing in the dual $\AdS_4 \times S^7/\Gamma$ Freund-Rubin background.

The isometry group of $S^7$ is $\SO(8)$, whence in particular any subgroup $\Gamma < \SO(8)$ acts linearly on $\RR^8$ and we can consider the orbifold $\RR^8/\Gamma$.  Since linear transformations fix the origin, this is always an orbifold even if $S^7/\Gamma$ is smooth. One can of course think of $\RR^8$ here as either $\CC^4$ or $\HH^2$. In terms of matter hypermultiplets, the quaternionic perspective is more natural since it makes the $\eN = 4$ structure manifest and indeed the finite subgroups of $\Sp(1) \times \Sp(1)$ act naturally on $\HH \oplus \HH$ via left multiplication. In the existing literature however it is often preferred to describe moduli spaces in complex notation, by thinking of each matter hypermultiplet as consisting of a pair of chiral supermultiplets.  The action of $\Gamma$ on $\CC^4$ is generally not complex-linear though this notation does make sense for the $\eN\geq 4$ orbifolds discussed here where the action is in fact $\HH$-linear and thus $\CC$-linear. The translation between the quaternionic notation we use and the complex notation in the rest of the literature is as follows.

The action of the finite subgroups of $\Sp(1) \times \Sp(1)$ on $\HH \oplus \HH$ via left multiplication is $\HH$-linear provided that scalar multiplication acts on $\HH \oplus \HH$ on the right.  Associativity of quaternion multiplication guarantees that the action of $\Sp(1) \times \Sp(1)$ commutes with scalar multiplication. We give $\HH \oplus \HH$ the structure of a complex vector space by restricting scalars from $\HH$ to $\CC$.  This sets an isomorphism $\HH \cong \CC^2$ and hence $\HH \oplus \HH \cong \CC^4$ as follows.  It is clearly enough to identify $\HH$ with $\CC^2$.  To do this we identify the quaternion $x = z + j w \in \HH$, where $z,w \in \CC$, with the vector $(z,w) \in \CC^2$.  This identification is complex linear: for all $\lambda \in \CC$, we have that
\begin{equation}
  \lambda x := (z+jw)\lambda = \lambda z + j \lambda w \rightsquigarrow (\lambda z, \lambda w) = \lambda(z,w)~.
\end{equation}
Now let $u = a + bj \in \Sp(1)$, so that $a,b\in\CC$ and $|a|^2 + |b|^2 = 1$.  Then $ux =(a + bj)(z+jw) = (az - bw) + j(\bbar z + \abar w)$, whence
\begin{equation}
  a+ bj \mapsto
  \begin{pmatrix}
    a & -b \\ \bbar & \hphantom{-}\abar
  \end{pmatrix}.
\end{equation}
Let us now conclude by illustrating this discussion by recovering all the known examples of $\eN \geq 4$ quotients in the M2-brane literature. They are either smooth with $\eN >4$ or cyclic orbifolds with $\eN =4$.   

\subsection{Smooth quotients}
\label{sec:smooth-quotients-2}

Consider first the $\ZZ_k$ subgroup of $\Sp(1) \times \Sp(1)$ generated by $g=(\omega_k, \omega_k^r)$, where $\omega_k = e^{2\pi i/k}$ is a primitive $k$th root of unity and $r$ is some integer coprime to $k$.  Then if $(x,y) \in \HH \oplus \HH$, we have that $g \cdot (x,y)  = (e^{2\pi i/k}x, e^{2\pi ir/k}y)$.  Writing $x = z_1 + j z_2$ and $y = z_3 + j z_4$, we find that the action of $g$ on $\CC^4$ is given by
\begin{equation}
 g\cdot (z_1,z_2,z_3,z_4) = (\omega_k z_1, \omega_k^{-1} z_2, \omega_k^r z_3, \omega_k^{-r} z_4)~.
\end{equation}
The orbifold $\CC^4 / \ZZ_k$ corresponds to the cone over the smooth cyclic quotient in the first row of Table~\ref{tab:smoothiterated}. Generically it has $\eN = 4$, with $\eN =6$ only if $r= \pm 1$ and $\eN =8$ only if $r= \pm1$ and $k=1$ or $k=2$. The untwisted case with $r= \pm 1$ arises as the dual geometry for the class of $\eN =6$ superconformal field theories in \cite{MaldacenaBL} with the Chern--Simons level identified with $k$. The other twisted smooth $\eN =4$ cyclic quotients with $r \neq \pm 1$ have no known superconformal field theory duals.  

Next, begin by considering again the example above but for the $\ZZ_{2k}$ subgroup of $\Sp(1) \times \Sp(1)$ generated by $g=(\omega_{2k}, \omega_{2k}^r)$, with $r$ now coprime to $2k$.  The subgroup $2 D_{2k}$ of $\Sp(1) \times \Sp(1)$ is generated by $g$ and $h = (j,j)$. The action of $g$ and $h$ on $\CC^4$ is given by
\begin{equation}
  \begin{aligned}[m]
  g\cdot (z_1,z_2,z_3,z_4) &= (\omega_{2k} z_1, \omega_{2k}^{-1} z_2, \omega_{2k}^r z_3, \omega_{2k}^{-r} z_4) \\
  h\cdot (z_1,z_2,z_3,z_4) &= (-z_2,z_1,-z_4,z_3)~. 
  \end{aligned}~
\end{equation}
The orbifold $\CC^4 / 2 D_{2k}$ corresponds to the cone over the smooth binary dihedral quotient in the second row of Table~\ref{tab:smoothiterated} obtained by quotienting first by the action of $g \in \ZZ_{2k}$ and then by $h \in \ZZ_2$. Generically it has $\eN = 4$ with $\eN =5$ only if $r= \pm 1$. The untwisted case with $r= \pm 1$ arises as the dual geometry for the class of $\eN =5$ superconformal field theories in \cite{3Lee,ABJ}. The extra quotient by $\ZZ_2$ in the dual geometry here has a direct interpretation from orientifolding  an $\eN =6$ theory in \cite{MaldacenaBL} with even Chern--Simons level $2k$ to obtain this class of $\eN =5$ theories. The twisted smooth $\eN =4$ binary dihedral quotients with $r \neq \pm 1$ have no known superconformal field theory duals. Neither do any of the smooth binary polyhedral quotients, be they untwisted with $\eN =5$ or twisted with $\eN =4$. 
  
\subsection{Cyclic orbifolds}
\label{sec:cyclic-orbifolds}

Consider first the product subgroup $\ZZ_p \times \ZZ_q$ of $\Sp(1) \times \Sp(1)$, with typical element $g=(\omega_p^a, \omega_q^b)$ for any $0 \leq a < p$ and $0 \leq b < q$.  Its action on $\CC^4$ is given by
\begin{equation}
  g \cdot (z_1,z_2,z_3,z_4) = (\omega_p^a z_1, \omega_p^{-a} z_2, \omega_q^b z_3, \omega_q^{-b} z_4)~.
\end{equation}
The orbifold $\CC^4 / \ZZ_p \times \ZZ_q$ corresponds to the cone over the product orbifold in the first row of Table~\ref{tab:productsiterated} (from which the chiral cyclic orbifold in Section~\ref{sec:orbif-prod-groups} follows as a special case if either $p$ or $q$ equal one). Orbifolds of this type were considered in \cite{Terashima:2008ba} for which a dual $\eN =4$ superconformal field theory was proposed as \lq orbifold gauge theory II' in the chiral case.  

The fibred product subgroup $\ZZ_{pk} \times_{\ZZ_k} \ZZ_{qk}$ of $\Sp(1) \times \Sp(1)$ is generated by multiplying the elements in $\ZZ_p \times \ZZ_q$ above with $g=(\omega_{pk}, \omega_{qk}^r)$, where $r$ is some integer coprime to $k$. The action of this extra generator on $\CC^4$ is given by
\begin{equation}
  g\cdot (z_1,z_2,z_3,z_4) = (\omega_{pk} z_1, \omega_{pk}^{-1} z_2, \omega_{qk}^r z_3, \omega_{qk}^{-r} z_4)~.
\end{equation}
The orbifold $\CC^4 / \ZZ_{pk} \times_{\ZZ_k} \ZZ_{qk}$ corresponds to the cone over the orbifold in the first row of Table~\ref{tab:remainingiterated} (from which the product orbifold above follows as a special case if $k=1$). In the untwisted case with $r= \pm 1$, orbifolds of this type have been obtained in \cite{MasahitoBL,KlebanovBL,Terashima:2008ba,Imamura:2008nn,Imamura:2008ji,Imamura:2009ur,Imamura:2009ph}. 

More precisely, the ones obtained in \cite{KlebanovBL,Terashima:2008ba} correspond to a special case of this type of orbifold wherein both $p$ and $q$ equal some fixed positive integer $n$ which is assumed to be coprime to $k$ (the one in \cite{MasahitoBL} corresponds to the case where $n=2$
\footnote{In this reference, a non-generic branch of the moduli space is also found which takes the form $\RR^8 / D_{2m}$ in terms of the ordinary dihedral group $D_{2m} = \ZZ_m \rtimes \ZZ_2$, though this orbifold preserves only $\eN =3$ supersymmetry.}
). In this special case, $\ZZ_n \times \ZZ_k \cong \ZZ_{nk}$ and the results in Appendix~\ref{sec:struct-fibr-prod} imply that $\ZZ_{nk} \times_{\ZZ_k} \ZZ_{nk} \cong \ZZ_{nk} \times \ZZ_n$ which is how the quotient is referred to in these references. The $\eN =4$ superconformal field theories proposed to be the holographic duals of such orbifolds are referred to as the \lq non-chiral orbifold gauge theory' in \cite{KlebanovBL} and \lq orbifold gauge theory I' in \cite{Terashima:2008ba} and correspond to a special case of the class of $\eN =4$ superconformal field theories first obtained in \cite{pre3Lee} such that the gauge group consists of a product of an even number of unitary groups. The gauge-matter couplings for these theories are encoded by quiver diagrams consisting of an even number of $2n$ nodes connected by $2n$ matter hypermultiplets to form a chain. The Chern--Simons level for each node is $\pm k$ with the sign alternating from node to node as one traverses the chain. Matter hypermultiplets must also alternate between twisted and untwisted type, so that there are $n$ of each type (see \cite{pre3Lee,SCCS3Algs} for more details). It is worth remarking that $S^7 / \ZZ_{nk} \times_{\ZZ_k} \ZZ_{nk}$ can be written as either $( S^7 / \ZZ_{n} \times \ZZ_{n} ) / \ZZ_k$ or $S^7 / \ZZ_{nk} / \ZZ_n$. In the latter form, the extra quotient by $\ZZ_n$ in the dual geometry has a direct interpretation from orbifolding an $\eN =6$ theory in \cite{MaldacenaBL} with Chern--Simons level $nk$, which is how this class of $\eN =4$ theories was obtained in \cite{KlebanovBL}. 

The general orbifolds of this type were obtained in \cite{Imamura:2008nn,Imamura:2008ji,Imamura:2009ur,Imamura:2009ph} and the class of $\eN =4$ superconformal field theories proposed as their holographic duals are referred to as \lq elliptic models'. Like the theories in \cite{KlebanovBL,Terashima:2008ba}, these elliptic models also have gauge-matter couplings encoded by chain quivers only now the number of nodes is given by $p+q$, where $p$ and $q$ correspond respectively to the numbers of, say, untwisted and twisted hypermultiplets forming the links in the chain. In particular, $p$ need not equal $q$ here and so the number of nodes need not be even. The Chern--Simons level for each node in the quiver is either $\pm k$ or zero with a zero occurring at each node whose pair of connecting links are hypermultiplets without a relative twist.

In the twisted case with $r \neq \pm 1$, none of the cyclic orbifolds have known $\eN =4$ superconformal field theory duals, nor indeed, to the best of our knowledge, do any of the remaining nonabelian $\eN =4$ orbifolds in Table~\ref{tab:remainingiterated}. It is hoped that progress in this direction may be aided by our description of these new orbifolds in terms of iterated quotients in order to perhaps obtain dual $\eN =4$ superconformal field theories for some of them via a projection of the known theories (e.g. via further orbifolding or orientifolding), perhaps along the lines discussed in \cite{Berenstein:2009ay}.

\subsection{M5-brane orbifolds}
\label{sec:m5-branes}

Let us mention briefly that a similar, but much simpler situation obtains with M5-branes, as discussed, for example, in \cite[§5.2]{AFHS}.

In this case, we are interested in supersymmetric backgrounds $\AdS_7 \times X^4$, with $X$ possibly an orbifold. Bär's construction, together with the non-existence of irreducible five-dimensional holonomy representations, imply that the only (complete) four-dimensional manifold admitting real Killing spinors is the round sphere $S^4$, hence any other supersymmetric background must be an orbifold of $S^4$ by a finite subgroup of $\SO(5)$, lifting isometrically to a subgroup $\Gamma < \Sp(2)$.  The space of Killing spinors is again identified with the $\Gamma$-invariant parallel spinors on $\RR^5$ which is the irreducible spinor representation $\Delta$ of $\Sp(2)$.  This representation is quaternionic, whence the space of $\Gamma$-invariant spinors is a quaternionic subspace: if a spinor is invariant, so is its quaternion line, by the quaternion-linearity of the action of $\Sp(2)$.  Since $\dim_\HH\Delta =2$, necessarily $0\leq \dim_\HH\Delta^\Gamma \leq 2$.  Hence if we demand some supersymmetry, either $\Gamma = \{1\}$ and we have $X = S^4$, or else the orbifold is half-BPS.  In this case $\Gamma$ is contained in an $\Sp(1)$ subgroup of $\Sp(2)$, leaving a nonzero vector invariant in the fundamental representation of $\Sp(2)$.  Up to automorphisms, we see that $\Gamma$ is one of the ADE subgroups in Table \ref{tab:ADE}, but this time embedded in $\Sp(2)$ in such a way that if $u \in \Gamma < \Sp(1)$ and $(x,y)\in\HH^2$, then $u \cdot (x,y) = (ux,y)$.

The action of that $\Sp(1)$ subgroup of $\SO(5)$ on $S^4$ is given by restricting the action on $\RR^5$.  This is given as follows.  First of all, there is a vector which is fixed, call it $\bv$.  If we identify the four-dimensional subspace perpendicular to $\bv$ with $\HH$, then the action of $\Sp(1)$ is by left quaternion multiplication.  Finally, using the arguments described in Section \ref{sec:orbi-iter-quots} and in particular Table~\ref{tab:subnormal}, it is a simple exercise to decompose such orbifolds $S^4/\Gamma$, except for $\Gamma = \sE_8$, into a sequence of cyclic quotients.  This should become useful if and when we understand the six-dimensional superconformal field theory dual to $\AdS_7 \times S^4$.

\section*{Acknowledgments}

This work was supported in part by grant ST/G000514/1 ``String Theory Scotland'' from the UK Science and Technology Facilities Council.  We are grateful to Chris Smyth for correspondence and to the MathOverflow community and, in particular, Keith Conrad, Robin Chapman and Theo Johnson-Freyd for their generous help in answering questions related to the research in this paper.

\appendix

\section{Structure of fibred products}
\label{sec:struct-fibr-prod}

In this appendix we show that even though the finite subgroups $A\times_{(F,\tau)} B$ of $\Spin(4)$ corresponding to different $\tau \in \Twist(F)$ are generally not conjugate in $\Spin(4)$, they are (in almost most cases) \emph{abstractly} isomorphic as groups.  The independence on the automorphism is of course trivially true in those cases where $\Twist(F)$ is a singleton, but we will now see that this is the case also in those cases where $\Twist(F) \neq \{1\}$, with two possible exceptions.   In all cases where we can show this, the result follows from Lemma~\ref{lem:autolift} below after exhibiting lifts of every $\tau \in \Twist(F)$ to $\Aut(A)$ or $\Aut(B)$ and in many cases the lift follows from Lemma~\ref{lem:cycliclift} below.  This does not seem to be totally trivial, in that outer automorphisms generally do not lift and, furthermore, there are examples of fibred products with inequivalent twisting automorphisms which are not abstractly isomorphic.  This appendix owes a lot to the collective wisdom of the MathOverflow community and in particular to the answers provided by the users mentioned in the acknowledgments to some of the questions asked by the senior author.

We start with two preliminary results.

\begin{lemma}
  \label{lem:autolift}
  If $\tau \in \Twist(F)$ is induced from an automorphism of either $A$ or $B$, then $A \times_{(F,\tau)} B \cong A \times_{(F,\id)} B$.
\end{lemma}

\begin{proof}
  Let us assume without loss of generality that $\tau$ is induced from $\that \in \Aut(A)$.  (The case where it is induced from an automorphism of $B$ is treated similarly.)  Let $\alpha:A \to F$ and $\beta: B \to F$ be the homomorphisms defined in equations \eqref{eq:alpha} and \eqref{eq:beta}, respectively.  Then this means that $\alpha(\that a) = \tau \alpha(a)$ for all $a \in A$.  We define an isomorphism $\varphi: A \times_{(F,\id)} B \to A \times_{(F,\tau)} B$ as follows.  Let $(a,b) \in A\times_{(F,\id)} B$ and let $\varphi(a,b) = (\that a, b)$.  Then since $\alpha(a) = \beta(b)$, we see that $\alpha(\that a) = \tau \alpha(a) = \tau \beta(b)$, whence indeed $(\that a , b) \in A \times_{(F,\tau)} B$.  Since $\that$ is an automorphism, $\varphi$ is a group isomorphism.
\end{proof}

A special case we will have ample opportunity to use is that of cyclic groups, for which automorphisms always lift.

\begin{lemma}
  \label{lem:cycliclift}
  Let $\ZZ_{mn} \to \ZZ_n$ be a group homomorphism.   Then if $r \in \ZZ_n^\times$, there exists some $s \in \ZZ_{mn}^\times$ with $s \equiv r \mod n$.
\end{lemma}

\begin{proof}
  (We learnt this proof from Keith Conrad at MathOverflow \cite{MO32878}.)  Let us first consider the special case where $m$ and $n$ are coprime.  Then by the Chinese Remainder Theorem, there is a unique $s \in\ZZ_{mn}$ solving the congruences $s \equiv r \mod n$ and $s \equiv 1 \mod m$.  The first congruence says that $s$ is coprime to $n$ since $r$ is, whereas the second congruence says that $s$ is coprime to $m$, hence $s \in \ZZ_{mn}^\times$.  In the general case, let $m$ and $n$ have greatest common divisor $\ell$, so that $m = \ell m'$ and $n = \ell p$ with $m'$, $\ell$ and $p$ pairwise coprime.  Let $n' = \ell^2 p$, so that $m'n'=mn$, but now $m'$ and $n'$ are coprime. Since $r$ is coprime to $n$, it is also coprime to $n'$ and hence, by the special case, there exists a unique $s \in \ZZ^\times_{mn}$ such that $s \equiv r \mod {n'}$.  But then $s \equiv r \mod n$ as well.
\end{proof}

We now discuss the different groups in some detail.

\begin{itemize}
\item $\Gamma = \sA_{kl-1} \times_{(\ZZ_l,\tau)} \sA_{ml-1}$\\
  Here $\tau$ is represented by some $r \in \ZZ_l^\times$; that is, an integer $r$ coprime to $l$.  Then by Lemma \ref{lem:cycliclift}, there exists an integer $s$ congruent to $r$ modulo $l$ such that $s$ is coprime to $kl$.  In other words, $s \in \ZZ_{kl}^\times$ defines an automorphism $\that$ of $\ZZ_{kl}$ which lifts $\tau$.  By Lemma~\ref{lem:autolift}, the isomorphism type of $\Gamma$ does not depend on $\tau$.  Let us therefore take $\tau = \id$.  Then $\Gamma$ is the extension
  \begin{equation}
    \begin{CD}
      1 @>>> \sA_{k-1} \times \sA_{m-1} @>>> \Gamma @>>> \ZZ_l @>>> 1~.
    \end{CD}
  \end{equation}
Let us consider the element $\gamma = (\omega_{kl},\omega_{ml}) \in \Gamma$, which is sent to the generator of $\ZZ_l$.  The group $\left<\gamma\right>$ generated by $\gamma$ has order $al$, where $a$ is the least common multiple of $k$ and $m$. In fact, it is clear that $al$ is an exponent of $\Gamma$, whence $\left<\gamma\right>$ is a summand of $\Gamma$.  Since $\Gamma$ is covered by $\ZZ^2$, it is isomorphic to the direct product of (at most) two cyclic groups, whence counting order $\Gamma \cong \ZZ_{al} \times \ZZ_b$, where $b$ is the greatest common divisor of $k$ and $m$.  (We learnt this proof from Robin Chapman over at MathOverflow \cite{MO30656}.)

\item $\Gamma=\sA_{3k-1} \times_{(\ZZ_3,\tau)} \sE_6$\\
A similar argument shows that this subgroup is isomorphic to $\sA_{3k-1} \times_{(\ZZ_3,\id)} \sE_6$, since the nontrivial outer automorphism of $\ZZ_3$ lifts to an automorphism of $\ZZ_{3k}$, by Lemma \ref{lem:cycliclift}.

\item $\Gamma = \sD_{k(2l+1)+2} \times_{(D_{4l+2},\tau)} \sD_{m(2l+1)+2}$\\
  The outer automorphism group of $D_{4l+2}$ is isomorphic to $\ZZ_{2l+1}^\times/\left<-1\right>$.  Let us concentrate on the surjection $\sD_{k(2l+1)+2} \to D_{4l+2}$, sending $x$ to $[x]$, where $[x]$ is the coset of $x$ relative to the normal subgroup generated by $t^{2l+1}$.   Then $r \in \ZZ_{2l+1}^\times$ represent an outer automorphism $\tau$ of $D_{4l+2}$.  Its action on $D_{4l+2}$ is given by
  \begin{equation}
    \tau[t^p] = [t^{rp}] \qquad\text{and}\qquad \tau [st^p] = [s t^{rp}]~.
  \end{equation}
  By Lemma~\ref{lem:cycliclift}, there exists $r' \in \ZZ_{2k(2l+1)}^\times$ with $r' \equiv r \mod{2l+1}$ and this in turn defines an automorphism $\that$ of $\sD_{k(2l+1)+2}$ defined by
  \begin{equation}
    \that(t^p) = t^{r'p} \qquad\text{and}\qquad \that(st^p) = st^{r'p}~.
  \end{equation}
  Notice, though, that $[t^{r'p}] = [t^{rp}]$ since $r' \equiv r \mod{2l+1}$ and $[t^{2l+1}]=1$.  Therefore $[\that x] =\tau[x]$ for all $x \in \sD_{k(2l+1)+2}$ and hence $\that$ is the lift of $\tau$.  Finally, by Lemma~\ref{lem:autolift}, the isomorphism type of $\Gamma$ is independent of $\tau$.  In those cases when the factor group is isomorphic to either $D_4$ or $2D_4$, not every automorphism of the factor group lifts, and hence Lemma~\ref{lem:autolift} is not applicable.  Such cases require a more detailed analysis which is beyond the scope of this paper.

\item $\Gamma=\sD_{2kl+2} \times_{(D_{4l},\tau)} \sD_{2ml+2}$ ($l\neq 1$) and $\Gamma=\sD_{l(2k+1)+2} \times_{(2D_{2l},\tau)} \sD_{l(2m+1)+2}$ ($l\neq 2$)\\
  These two cases are very similar.  In both cases, there are two kinds of nontrivial outer automorphisms.  The ones in $\ZZ_{2l}^\times$ lifts just as in the previous example and we will not discuss this further.  Let $\tau$ denote the automorphism of $D_{4l}$ corresponding to the nontrivial element in the $\ZZ_2$-factor of $\Out(D_{4l})$.  It is defined by
  \begin{equation}
    \tau[t^p] = [t^p] \qquad\text{and}\qquad \tau[st^p] = [st^{p+1}]~.
  \end{equation}
  Let $\that$ be the automorphism of $\sD_{2kl+2}$ defined by
  \begin{equation}
    \that(t^p) = t^p \qquad\text{and}\qquad \that(st^p) = (st^{p+1})~.
  \end{equation}
  Then $\tau[x] = [\that x]$ and hence $\tau$ lifts to an automorphism of $\sD_{2kl+2}$.  The same argument, \emph{mutatis mutandis}, shows that for $\sD_{l(2k+1)+2} \times_{(2D_{2l},\tau)} \sD_{l(2m+1)+2}$, the outer automorphisms of $2D_{2l}$ lift to automorphisms of $\sD_{l(2k+1)+2}$.  In both cases, Lemma~\ref{lem:autolift} implies that the fibred products are (up to isomorphism) independent of the twisting automorphism.

\item $\Gamma =\sE_6 \times_{(T,\tau)} \sE_6$\\
  Here $\tau$ is the unique nontrivial outer automorphism of the tetrahedral group $T$, which is induced from the unique nontrivial outer automorphism $\that$ of $\sE_6$.  We let $x \mapsto [x]$ denote the map $\sE_6 \to T$ and let $Z$ denote the kernel of this map, which is the centre of $\sE_6$.  In terms of quaternions it is the subgroup $\{\pm 1\}$.  We claim that the automorphism $\tau [x] := [\that x]$ is not inner.  Indeed, suppose that $\tau[x] = [z][x][z]^{-1}$ for some $[z]\in T$.  Then $[\that x] = [z x z^{-1}]$, whence
  \begin{equation}
    \that x = \varepsilon(x) z x z^{-1}~,
  \end{equation}
  where $\varepsilon: \sE_6 \to Z$ is a group homomorphism (from the fact that $\that$ is an automorphism).  The kernel of this homomorphism is either all of $\sE_6$ or else a normal subgroup of index 2.  However as seen in Section \ref{sec:binary-tetr-group}, $\sE_6$ has no such normal subgroups.  This means that $\varepsilon(x) = 1$ for all $x$ and hence that $\that = z x z^{-1}$, contradicting the fact that $\that$ is not inner.  This means that $\tau \in \Out(T)$ lifts and, by Lemma \ref{lem:autolift}, we  conclude that $\Gamma \cong \sE_6 \times_{(T,\id)} \sE_6$.  We can determine the structure of this group as follows.  It consists of the following elements of $\sE_6 \times \sE_6$:
  \begin{equation}
    \Gamma = \left\{(a,a)\middle | a \in \sE_6 \right\} \cup \left\{(a,-a)\middle | a\in\sE_6 \right\}~.
  \end{equation}
  The diagonal subgroup $\left\{(a,a)\middle | a \in \sE_6 \right\}$ has index 2 and is hence normal.  This means we have an exact sequence
  \begin{equation}
    \begin{CD}
      1 @>>> \sE_6 @>>> \Gamma @>>> \ZZ_2 @>>> 1~,
    \end{CD}
  \end{equation}
  which is easily seen to split, with the homomorphism $\ZZ_2 \to \Gamma$ given by sending the generator $-1$ to $(1,-1)$.  Since $(1,-1)$ is central, we see that $\Gamma \cong \sE_6 \times \ZZ_2$.

\item $\Gamma=\sE_7 \times_{(O,\id)} \sE_7$\\
  Here there are no nontrivial outer automorphisms and the same argument, \emph{mutatis mutandis}, as in the previous case shows that $\sE_7 \times_O \sE_7 \cong \sE_7\times \ZZ_2$.

\item $\Gamma=\sE_8 \times_{(I,\tau)} \sE_8$\\
  The same argument, \emph{mutatis mutandis}, as in the case $\sE_6 \times_{(T,\tau)} \sE_6$ shows that for $\tau$ the nontrivial element of $\Out(I)$, $\sE_8 \times_{(I,\tau)} \sE_8 \cong \sE_8 \times_{(I,\id)} \sE_8 \cong \sE_8 \times \ZZ_2$.
\end{itemize}

The last three cases imply that the corresponding orbifolds $S^7/\Gamma$ are $\ZZ_2$-orbifolds of the smooth $\eN=4$ and $\eN=5$ quotients associated to the ADE subgroups $\sE_{6,7,8}$ and classified in \cite{deMedeiros:2009pp}.

\section{Finite subgroups of $\SO(4)$}
\label{sec:finite-subgroups-so4}

In this appendix will summarise an independent check of our classification of finite subgroups of $\Spin(4)$ (up to conjugation) by showing that we recover the classification of finite subgroups of $\SO(4)$ in \cite[§4]{MR1957212}, particularly their Tables~4.1 and 4.2.

\subsection{Notation}
\label{sec:notation}

The notation in \cite{MR1957212} deserves some comment.  First of all, they call subgroups of $\SO(4)$ \emph{chiral}, to distinguish them from the \emph{achiral} subgroups of the general orthogonal group. (This is not to be confused with the notion of chiral subgroup introduced in Section \ref{sec:orbif-prod-groups}.)  Chiral subgroups are further divided into \emph{diploid} and \emph{haploid} subgroups, according to whether or not the subgroup contains the orthogonal transformation $-\id \in \SO(4)$, sending $x$ to $-x$, which is chiral in four dimensions.

The notation for the 2-to-1 covering homomorphism $\Sp(1) \times \Sp(1) \to \SO(4)$ is such that $(l,r) \mapsto [l,r]$, for $l,r\in\Sp(1)$.  The kernel of this homomorphism is the order-2 subgroup generated by $(-1,-1)$, whence $[-l,-r]=[l,r]$.  In this notation, the orthogonal transformation $-\id$ is denoted $[1,-1]$ or equivalently $[-1,1]$.

Haploid subgroups of $\SO(4)$ are denoted $+\frac1{f} [L \times R]$, where $L,R < \SO(3)$ and $f$ is the order of the relevant factor group (as in Goursat's Lemma).  Up to at most a dichotomy, the order determines the factor group, whence the notation is usually not ambiguous.  In case of ambiguity, either $L$ or $R$ are further adorned with a bar.  Diploid subgroups of $\SO(4)$ are denoted $\pm \frac1{f}[L\times R]$ with similar meanings to the symbols.  The $\pm$ is appropriate because if $g$ belongs to a diploid subgroup, so does $-g$ (with $\pm g$ thought of as $4\times 4$ matrices).  The order of a haploid subgroup $+\frac1{f} [L \times R]$ is given by $|L||R|/f$, whereas that of a diploid subgroup $\pm \frac1{f}[L\times R]$ is twice that: $2|L||R|/f$.

The finite subgroups of $\SO(4)$ classified in \cite{MR1957212} are tabulated in Tables 4.1 and 4.2 in that paper.  Next we will recover this classification from our classification of finite subgroups of $\Spin(4)$ and will exhibit the precise correspondence between subgroups.  We believe this provides an independent check, both of our results and, if necessary, of those in \cite{MR1957212}.

\subsection{Recovering the classification}
\label{sec:recov-class}

As the notation described above makes clear, the classification of finite subgroups of $\SO(4)$ is based on the classification of finite subgroups of $\PSO(4) \cong \SO(3)\times\SO(3)$.  The groups $\Spin(4)$, $\SO(4)$ and $\PSO(4)$ are related as follows.  The centre of $\Spin(4)$ is isomorphic to $\ZZ_2 \times \ZZ_2$.  Under the isomorphism $\Spin(4) \cong \Sp(1)\times \Sp(1)$, the centre is the subgroup $Z$ consisting of the four elements: $(1,1)$, $(1,-1)$, $(-1,1)$ and $(-1,-1)$.  The kernel of the homomorphism $\Sp(1)\times\Sp(1) \to \SO(4)$ is the subgroup generated by $(-1,-1)$.  The two elements $(1,-1)$ and $(-1,1)$ map to the same element $[-1,1]$ of $\SO(4)$: namely, $-\id$.  This element generates the kernel of the homomorphism $\SO(4) \to \SO(3)\times\SO(3)$.  Let us introduce the notation $\pi: \Sp(1) \to \SO(3)$ for the covering homomorphism.  The kernel of $\pi$ is the order-2 subgroup generated by $-1$.  Restricted to the ADE subgroups of $\Sp(1)$, we have the following correspondence:
\begin{equation}
  \begin{array}{c|cccccc}
    G & \sA_{2n-1} & \sA_{2n} & \sD_{k+2} & \sE_6 & \sE_7 & \sE_8 \\\hline
    \pi(G) & C_n & C_{2n+1} & D_{2k} & T & O & I
  \end{array}
\end{equation}
where $C_n$ denotes the cyclic group of order $n$, and where $\pi$ is a double cover in all cases but $\sA_{2n} \to C_{2n+1}$, where it is an isomorphism.

Now let $\Gamma < \Sp(1)\times \Sp(1)$, whence $\Gamma = A \times_{(F,\tau)}B$, where $A,B$ are finite subgroups of $\Sp(1)$ with common factor $F$ and $\tau \in \Aut(F)$.  Equivalently, $\Gamma$ is a categorical pull-back (with all maps epimorphisms)
\begin{equation}
  \label{eq:pbdiag}
  \xymatrix{\Gamma \ar@{->}[r]^\rho \ar@{->}[d]_\lambda & B \ar@{->}[d]^\beta\\
  A \ar@{->}[r]^\alpha & F}
\end{equation}
where $\alpha,\beta$ incorporate the automorphism $\tau \in \Aut(F)$.  Let $\Gbar$ denote the projection of $\Gamma$ to $\SO(3) \times \SO(3)$.  Let
$\Abar$ and $\Bbar$ be the images of $A$ and $B$, respectively, under $\pi$.  Since $\Gamma < A \times B$, it follows that $\Gbar < \Abar \times \Bbar$.  This gives maps $\lambdabar : \Gbar \to \Abar$ and $\rhobar : \Gbar \to \Bbar$ making the following diagram commute:
\begin{equation}
  \label{eq:pbdiagbar}
  \xymatrix{ & \Gamma \ar@{->}[dl]_\lambda \ar@{->}[dd] \ar@{->}[dr]^\rho & \\
    A \ar@{->}[dd]_\pi & & B \ar@{->}[dd]^\pi \\
    & \Gbar \ar@{->}[dl]_\lambdabar \ar@{->}[dr]^\rhobar & \\
    \Abar & & \Bbar }
\end{equation}
It follows categorically that since $\pi,\lambda,\rho$ are epimorphisms, so are $\lambdabar$ and $\rhobar$.  By Goursat's Lemma, $\Gbar$ is also then given by a categorical pull-back (with all maps epimorphisms)
\begin{equation}
  \xymatrix{\Gbar \ar@{->}[r]^\rhobar \ar@{->}[d]_\lambdabar & \Bbar \ar@{->}[d]^\betabar\\
  \Abar \ar@{->}[r]^\alphabar & \Fbar}
\end{equation}
for some morphisms $\alphabar,\betabar$ to a common factor $\Fbar$.  Once given $\Abar$ and $\Bbar$, we determine $\alphabar:\Abar \to \Fbar$ and $\betabar: \Bbar \to \Fbar$ from the knowledge of $\Gbar$ as in the proof of Goursat's Lemma.  One thing we can say in general is that there is an epimorphism $\phi: F \to \Fbar$ in such a way that the following cube commutes:
\begin{equation}
  \xymatrix@!0{ & \overline\Gamma \ar@{->}[rr]\ar@{->}'[d][dd] & & \overline B \ar@{->}[dd]\\
    \Gamma \ar@{->}[ur]\ar@{->}[rr]\ar@{->}[dd] & & B \ar@{->}[ur]\ar@{->}[dd] \\
    & \overline A \ar@{->}'[r][rr] & & \overline F \\
    A \ar@{->}[rr]\ar@{->}[ur] & & F \ar@{.>}[ur]_\phi}
\end{equation}
where the front and back faces are the pull-back diagrams \eqref{eq:pbdiag} and \eqref{eq:pbdiagbar} and the three solid arrows between them are all $\pi$.   Indeed, let $f \in F$.  Then there is some $(l,r) \in \Gamma$ with $\alpha(l) = \beta(r) = f$.  We define $\phi(f) \in \Fbar$ by $\phi(f) = \alphabar(\overline l) = \betabar(\overline r)$.  One readily checks that $\phi$ is well-defined and again an epimorphism because so are $\pi,\alpha,\alphabar$ or $\pi,\beta,\betabar$.

We are actually interested in $[\Gamma]$, which is the image of $\Gamma$ under the covering homomorphism $\Sp(1) \times \Sp(1) \to \SO(4)$.  To understand the relationship between $[\Gamma]$ and $\Gbar$ we need to understand how $\Gamma$ interacts with the centre $Z$ of $\Sp(1) \times \Sp(1)$.  The lattice of subgroups of $Z$ is given by
\begin{equation}
  \xymatrix{& Z & \\
    \left<(1,-1)\right> \ar@{-}[ur] & \left<(-1,1)\right> \ar@{-}[u] & \left<(-1,-1)\right> \ar@{-}[ul] \\
    & \{(1,1)\} \ar@{-}[ul] \ar@{-}[u] \ar@{-}[ur]
  }
\end{equation}
whence there are four different possibilities for $\Gamma \cap Z$:
\begin{enumerate}\renewcommand{\labelenumi}{(\alph{enumi})}
\item $\Gamma \cap Z = Z$: in this case $\Gamma \not\cong [\Gamma] \not\cong \Gbar$, whence $[\Gamma]$ is diploid;
\item $\Gamma \cap Z = \left<(-1,1)\right>$ or $\left<(1,-1)\right>$: in this case $\Gamma \cong [\Gamma] \not\cong \Gbar$, whence $[\Gamma]$ is again diploid;
\item $\Gamma \cap Z = \left<(-1,-1)\right>$: in this case $\Gamma \not\cong [\Gamma] \cong \Gbar$, whence $[\Gamma]$ is haploid; and
\item $\Gamma \cap Z = \{(1,1)\}$: in this case $\Gamma \cong [\Gamma] \cong \Gbar$, whence $[\Gamma]$ is again haploid.
\end{enumerate}
In the first two cases we have that $[\Gamma] = \pm \frac{1}{f}[\Abar \times \Bbar]$, where $f = |\Fbar|$; whereas in the last two cases, $[\Gamma] = + \frac{1}{f}[\Abar \times \Bbar]$.

It is now a simple matter of going in turn through every single finite subgroup of $\Sp(1) \times \Sp(1)$ in Tables \ref{tab:products}, \ref{tab:smooth} and \ref{tab:remaining}, determining which case (a)-(d) obtains and the nature of the group $\Fbar$.  Doing so we recover Tables 4.1 and 4.2 in \cite{MR1957212} with one small correction: namely, the penultimate entry in Table~4.1, corresponding to the haploid subgroup $+\half [D_{2m} \times C_{2n}] $ is missing the condition that both $m$ and $n$ be odd, which clearly follows from their choice of generators.  The precise correspondence between the subgroups $\Gamma < \Sp(1)\times \Sp(1)$ and $[\Gamma]<\SO(4)$ is given in Tables~\ref{tab:prodSO4}, \ref{tab:smoothSO4} and \ref{tab:remainSO4} below.  The fact that the smooth subgroups give rise to haploid subgroups is easy to explain: a smooth subgroup $\Gamma$ of $\Sp(1) \times \Sp(1)$ is the graph of an automorphisms and automorphisms preserve the centre, whence $(-1,-1) \in \Gamma$, but $(\pm 1,\mp 1) \not\in\Gamma$.

\begin{table}[h!]
  \caption{Subgroups of $\SO(4)$ coming from subgroups in Table~\ref{tab:products} \label{tab:prodSO4}}
  \centering
  \begin{tabular}[t]{>{$}l<{$}|>{$}l<{$}}
    \multicolumn{1}{c|}{$\Gamma$} & \multicolumn{1}{c}{$[\Gamma]$}\\\hline
    \sA_{2n-1} \times \sA_{2m-1} & \pm [C_n \times C_m]\\
    \sA_{2n-1} \times \sA_{2m} & \pm [C_n \times C_{2m+1}]\\
    \sA_{2n} \times \sA_{2m} & + [C_{2n+1} \times C_{2m+1}]\\
    \sA_{2n-1} \times \sD_{m+2} & \pm [C_n \times D_{2m}]\\
    \sA_{2n} \times \sD_{m+2} & \pm [C_{2n+1}\times D_{2m}]\\
    \sA_{2n-1} \times \sE_6 & \pm [C_n \times T]\\
    \sA_{2n} \times \sE_6 & \pm [C_{2n+1}\times T]\\
    \sA_{2n-1} \times \sE_7 & \pm [C_n \times O]\\
    \sA_{2n} \times \sE_7 & \pm [C_{2n+1}\times O]\\
    \sA_{2n-1} \times \sE_8 & \pm [C_n \times I]\\
    \sA_{2n} \times \sE_8 & \pm [C_{2n+1}\times I]
  \end{tabular}
  \qquad\qquad
  \begin{tabular}[t]{>{$}l<{$}|>{$}l<{$}}
    \multicolumn{1}{c|}{$\Gamma$} & \multicolumn{1}{c}{$[\Gamma]$}\\\hline
    \sD_{n+2} \times \sD_{m+2} & \pm [D_{2n} \times D_{2m}]\\
    \sD_{n+2} \times \sE_6 & \pm [D_{2n} \times T]\\
    \sD_{n+2} \times \sE_7 & \pm [D_{2n} \times O]\\
    \sD_{n+2} \times \sE_8 & \pm [D_{2n} \times I]\\
    \sE_6 \times \sE_6 & \pm [T \times T]\\    
    \sE_6 \times \sE_7 & \pm [T \times O]\\    
    \sE_6 \times \sE_8 & \pm [T \times I]\\    
    \sE_7 \times \sE_7 & \pm [O \times O]\\    
    \sE_7 \times \sE_8 & \pm [O \times I]\\    
    \sE_8 \times \sE_8 & \pm [I \times I]
  \end{tabular}
\end{table}

\begin{table}[h!]
  \caption{Subgroups of $\SO(4)$ coming from subgroups in Table~\ref{tab:smooth} \label{tab:smoothSO4}}
  \centering
  \begin{tabular}[t]{>{$}l<{$}|>{$}l<{$}}
    \multicolumn{1}{c|}{$\Gamma$} & \multicolumn{1}{c}{$[\Gamma]$}\\\hline
    \sA_{2n-1} \times_{(\ZZ_{2n},\tau)} \sA_{2n-1} & +\frac1n [C_n \times C_n^{(s)}],~(s,n)=1\\
    \sA_{2n} \times_{(\ZZ_{2n+1},\tau)} \sA_{2n} & +\frac1{2n+1} [C_{2n+1} \times C_{2n+1}^{(s)}],~(s,2n+1)=1\\
    \sD_{n+2} \times_{(2D_{2n},\tau)} \sD_{n+2} & +\frac1{2n} [D_{2n} \times D_{2n}^{(s)}],~(s,2n)=1\\
    \sE_6 \times_{2T} \sE_6 & +\frac1{12} [T\times T]\\
    \sE_7 \times_{(2O,\tau)} \sE_7 & +\frac1{24} [O\times O]~\text{and}~+\frac1{24} [O\times \overline O]\\
    \sE_8 \times_{(2I,\tau)} \sE_8 & +\frac1{60} [I\times I]~\text{and}~+\frac1{60} [I\times \overline I]
 \end{tabular}
\end{table}

\begin{table}[h!]
  \caption{Subgroups of $\SO(4)$ coming from subgroups in Table~\ref{tab:remaining} \label{tab:remainSO4}}
  \centering
  \begin{tabular}[t]{>{$}l<{$}|>{$}l<{$}}
    \multicolumn{1}{c|}{$\Gamma$} & \multicolumn{1}{c}{$[\Gamma]$}\\\hline
    \sA_{2kl-1} \times_{(\ZZ_l,\tau)} \sA_{2ml-1} & \pm \frac1l [C_{kl} \times C_{ml}^{(s)}],~(s,l)=1\\
    \sA_{(2k+1)l-1} \times_{(\ZZ_l,\tau)} \sA_{2ml-1} & \pm \frac1l [C_{(2k+1)l} \times C_{ml}^{(s)}],~(s,l)=1\\
    \sA_{(2k+1)l-1} \times_{(\ZZ_l,\tau)} \sA_{(2m+1)l-1} & \pm \frac1l [C_{l(2k+1)} \times C_{l(2m+1)}^{(s)}],~(s,l)=1,~l\equiv1(2)\\
    \sA_{2l(2k+1)-1} \times_{(\ZZ_{2l},\tau)} \sA_{2l(2m+1)-1} & \pm \frac1l [C_{l(2k+1)} \times C_{l(2m+1)}^{(s)}],~(s,2l)=1\\
    \sA_{4k-1} \times_{\ZZ_2} \sD'_{2m+2} & \pm \frac12 [C_{2k} \times \Dbar_{4m}]\\
    \sA_{4k+1} \times_{\ZZ_2} \sD'_{2m+2} & \pm [C_{2k+1} \times \Dbar_{4m}]\\
    \sA_{4k-1} \times_{\ZZ_2} \sD_{m+2} & \pm \frac12 [C_{2k} \times D_{2m}]\\
    \sA_{4k+1} \times_{\ZZ_2} \sD_{m+2} & \pm [C_{2k+1} \times D_{2m}]\\
    \sA_{4k-1} \times_{\ZZ_4} \sD_{2m+3} & \pm \frac12 [C_{2k} \times D_{4m+2}],~k\equiv0(2)\\
    \sA_{4k-1} \times_{\ZZ_4} \sD_{2m+3} & + \frac12 [C_{2k} \times D_{4m+2}],~k\equiv1(2)\\
    \sA_{6k-1} \times_{\ZZ_3} \sE_6 & \pm \frac13 [C_{3k} \times T]\\
    \sA_{6k+2} \times_{\ZZ_3} \sE_6 & \pm \frac13 [C_{3(2k+1)} \times T]\\
    \sA_{4k-1} \times_{\ZZ_2} \sE_7 & \pm \frac12 [C_{2k} \times O]\\
    \sA_{4k+1} \times_{\ZZ_2} \sE_7 & \pm \frac12 [C_{4k+2} \times O]\\
    \sD'_{2k+2} \times_{\ZZ_2} \sD'_{2m+2} & \pm \frac12 [\Dbar_{4k} \times \Dbar_{4m}]\\
    \sD_{k+2} \times_{\ZZ_2} \sD_{m+2} & \pm \frac12 [D_{2k} \times D_{2m}]\\
    \sD_{2k+2} \times_{(D_4,\tau)} \sD_{2m+2} & \pm\frac14[D_{4k} \times D_{4m}]~\text{and}~\pm\frac14[D_{4k} \times \Dbar_{4m}]\\
    \sD_{lk+2} \times_{(D_{2l},\tau)} \sD_{lm+2} & \pm \frac1{2l} [D_{2lk} \times D_{2lm}^{(s)}]~,(s,l)=1,~l>2\\
    \sD_{2k+3} \times_{\ZZ_4} \sD_{2m+3} & + \frac12 [D_{4k+2} \times D_{4m+2}]\\
    \sD_{4k+4} \times_{(2D_4,\tau)} \sD_{4m+4} & +\frac14 [D_{8k+4} \times D_{8m+4}]~\text{and}~+\frac14 [D_{8k+4} \times \Dbar_{8m+4}]\\
    \sD_{l(2k+1)+2} \times_{(2D_{2l},\tau)} \sD_{l(2m+1)+2} & + \frac1{2l} [D_{2l(2k+1)} \times D_{2l(2m+1)}^{(s)}]~,(s,2l)=1,~l>2\\
    \sD_{k+2} \times_{\ZZ_2} \sD'_{2m+2} & \pm \frac12 [D_{2k} \times \Dbar_{4m}]\\
    \sD'_{2k+2} \times_{\ZZ_2} \sE_7 & \pm \frac12 [\Dbar_{4k} \times O]\\
    \sD_{k+2} \times_{\ZZ_2} \sE_7 & \pm \frac12 [D_{2k}\times O]\\
    \sD_{3k+2} \times_{D_6} \sE_7 & \pm \frac16 [D_{6k} \times O]\\
    \sE_6 \times_{\ZZ_3} \sE_6 & \pm \frac13 [T \times T]\\
    \sE_6 \times_{T} \sE_6 & \pm \frac1{12} [T \times T]\\
    \sE_7 \times_{\ZZ_2} \sE_7 & \pm \frac12 [O \times O]\\
    \sE_7 \times_{D_6} \sE_7 & \pm \frac16 [O \times O]\\
    \sE_7 \times_{O} \sE_7 & \pm \frac1{24} [O \times O]\\
   \sE_8 \times_{(I,\tau)} \sE_8 & \pm \frac1{60} [I \times I]~\text{and}~\pm \frac1{60} [I \times \overline I]
\end{tabular}
\end{table}

\bibliographystyle{utphys}
\bibliography{Sugra,Geometry,Algebra,AdS3}

\end{document}